\documentclass[12pt]{article}

\usepackage{xcolor}
\usepackage{colortbl}
\usepackage{amsmath}
\usepackage{amssymb}
\usepackage{amsthm}
\usepackage{graphicx}
\usepackage{epsfig}
\usepackage{multirow}
\usepackage[caption=false]{subfig}
\usepackage{algorithm}
\usepackage{algpseudocode}

\usepackage[outerbars,color]{changebar}
\usepackage{ulem}
\usepackage[normalsize,bf]{caption}

\providecolor{HSadded}{rgb}{0,0,1}
\providecolor{HSdeleted}{rgb}{0,2,2}
\providecolor{SCadded}{rgb}{3,0,2}
\providecolor{SCdeleted}{rgb}{0,2,0}

\newtheorem{thm}{Theorem}
\newtheorem{lem}[thm]{Lemma}
\newtheorem{corl}[thm]{Corollary}

\newcounter{saveissueFi}
\newcounter{saveissueMi}

\begin{document}
\title{A Design Methodology for Folded, Pipelined Architectures
in VLSI Applications using Projective Space Lattices}

\author{ \normalsize{Hrishikesh Sharma} ~~~~~~~~~~~~~~~~~~~~~~~~~~~~~~~~~~~~~~~~~~~~~~~~~~~~~~~ \normalsize{Sachin Patkar} \\
\normalsize{Department of Electrical Engg., Indian Institute of
Technology, Bombay, India}
}

\maketitle

\begin{abstract}
Semi-parallel, or folded, VLSI architectures are used whenever hardware
resources need to be saved at design time. Most recent applications that
are based on Projective Geometry\,(PG) based balanced bipartite graph also
fall in this category. In this paper, we provide a high-level, top-down
design methodology to design optimal semi-parallel architectures for
applications, whose Data Flow Graph\,(DFG) is based on PG bipartite graph.
Such applications have been found e.g. in error-control coding and matrix
computations. Unlike many other folding schemes, the topology of
connections between physical elements does not change in this methodology.
Another advantage is the ease of implementation. To lessen the throughput
loss due to folding, we also incorporate a \textit{multi-tier}
pipelining strategy in the design methodology. The design
methodology has been verified by implementing a synthesis tool in C++,
which has been verified as well. The tool is publicly available.
Further, a complete decoder was manually protototyped before the synthesis tool design, to
verify all the algorithms evolved in this paper, towards various steps of
refinement. Another
specific high-performance design of an LDPC decoder based on this
methodology was worked out in past, and has been patented as well.
\end{abstract}


\section{Introduction}
\label{intro_sec}

A number of naturally parallel computations make use of balanced bipartite
graphs arising from finite projective geometry \cite{hoholdt},
\cite{expanders}, \cite{mat_pap}, \cite{fossorier}, and related
structures \cite{rakov1}, \cite{rakov2}, \cite{nschau} to
represent their data flows. Many of them are in fact, recent
research directions, e.g. \cite{hoholdt}, \cite{rakov1},
\cite{fossorier}. As the dimension of the
finite projective space is increased, the corresponding graphs grow both in size
and order. Each vertex of the graph represents a LPU, and all
the vertices on one side of the graph \textbf{compute in parallel}, since there
are no data dependencies/edges between vertices that belong to one side of
a bipartite graph.  The number of such parallel LPUs is
generally of the order of tens of thousands in practice for various
reasons as noted below.

It is well-known in the area of error-control coding that higher the length
of error correction code, the closer it operates to Shannon limit of
capacity of a transmission channel \cite{fossorier}. The length of a code
corresponds to size of a particular bipartite graph, Tanner graph, which is
also the data flow graph for the decoding system \cite{fossorier}.
Similarly, in matrix computations, especially LU/Cholesky decomposition for
solving system of linear equations, and iterative PDE solving (and the
sparse matrix vector multiplication sub-problem within) using conjugate
gradient algorithm, the matrix sizes involved can be of similar high order.
A PG-based parallel data distribution can be imposed using suitable
interconnection of processors to provide \textbf{optimal} computation time
\cite{mat_pap}, which can result in quite big setup (as big as a petaflop
supercomputer).  This setup is being targeted in Computational Research
Labs, India, who are our collaboration partners. Further, at times,
\textbf{increasing} the dimension of finite projective geometry used in a
computation has been found to improve application performance
\cite{expanders}. In such a case, the number of LPUs grows
\textit{exponentially} with the dimension again. For \uline{practical} system
implementations with good application performance, it is
generally not possible to
have a large number of LPUs running in parallel, since that
incurs high manufacturing costs. In VLSI terms, such implementations may
suffer from relatively large area, and are also not scalable. Here,
\textit{scalability} captures the ease of using the same architecture for
extensions of the application that may require different throughputs, input
block sizes etc. A folded architecture can generally provide area reduction
and scalability as advantages instead, while trading off with system
throughput. We have therefore focused on designing \textbf{semi-parallel},
or folded architectures, for such PG-based applications.

The applicability of such schemes may not be that
widespread, given the current ULSI levels of integration. Still,
there are application areas in ASIC design, where direct
interconnect is still more pertinent (e.g., \cite{expanders} and
\cite{ldpc_pap}). This is because the required interconnect is a
\textit{sparse} interconnect. In fact, most practical designs reported here
are of semi-parallel nature. With such applications in mind, we a folding
scheme over next few sections.

As such, folding of VLSI architectures especially for communications and
signal processing systems is has been well-known \cite{parhi_book}.
However, the algorithms involved, such as register minimization algorithms,
are \textit{generic} in nature, and at times, iterative. We present
\textit{much
simpler set of algorithms} for folding for the target class of
applications.

In this paper, we \textbf{first} present a scheme for folding PG-based
computations efficiently, which allows a practical implementation with the
following advantages.

\begin{enumerate}
\item The number of on-chip \textit{logical} processing units required, is reduced.
\item No processing unit is ever idle in a machine cycle.
\item A schedule can be generated which ensures that there are no
        memory access conflicts between \textit{logical} processing units, for each
        (\textit{logical}) memory unit.
\item The same set of wires can be used to schedule communication of data
between memory units and processing units that are physically used across
multiple folds, \textbf{without} changing their interconnection.
\item Data distribution among the memories is such that the address
        generation circuits are simplified to counters/look-up tables.
\end{enumerate}

As an additional advantage of using this folding scheme, the
\textit{communication architecture} can be chosen to be
\textbf{point-to-point}. This is because same set of wires can be reused across
multiple folds, due to overlay (i.e., without reconfiguring their end points
at run time). This significantly reduces the amount of wiring resources
that are needed physically. Hence, a \textit{point-to-point} interconnection becomes
generally feasible after such folding. Such overlay-based custom
communication architecture leads to optimal performance, as will be brought
out in the paper. Generally, folding leads to overlay of computation, while
here, it \textit{simultaneously} leads to overlay of communication. Hence this
scheme can also be \textbf{alternatively} viewed as one of evolving
\textbf{custom} \textit{communication architecture}.

In general, custom communication architectures attempt to address the
shortcomings of standard on-chip communication architectures by utilizing
new topologies and protocols to obtain improvements for design goals, such
as performance and power. These novel topologies and protocols are often
\textit{customized} to suit a particular application, and
typically include optimizations to meet application-specific design goals.
In our case, the
foldable point-to-point communication is optimized towards PG-based
applications pointed out earlier.

This scheme forms the core of the \textit{design methodology} that is our main
contribution. The scheme is based on simple concepts of
modulo arithmetic, and circulance of PG-based balanced bipartite
graphs. It is an \textit{engineering-oriented},
practical alternative to another scheme based on vector space partitioning
\cite{dam_pap}. The core of that scheme is based on adapting the method of
vector space partitioning \cite{vs_part} to projective spaces, and
hence involves fair amount of mathematical rigor. A
\textbf{restricted} version of that scheme, which partitions the vector
space in a novel way, was worked out earlier using different methods
\cite{cacs_pap}. All this work was done as part of a research theme of
evolving \textit{optimal} folding architecture design methods, and also
applying such methods in real system design. As part of second goal, such
folding schemes have been used for design of specific decoder systems
having applications in secondary storage \cite{ldpc_foldpat}, \cite{expanders}.

The \textit{target} of this design methodology is to design specialized IP cores,
rather than a complete SoC. The methodology uses four levels of model
refinements. The level of details at these \textit{refinement levels} turn
out to be \textit{very similar} to the four levels in SpecC system-level
design methodology by Gajski et al \cite{spec_c_methodology}. Details of
this similarity are provided in section \ref{mod_ref_sec}. The latter
methodology was targeted for bus-based system designs. Still, the
similarity points to the fact that implementing a practical, custom
synthesis-based design flow for this methodology can indeed be worked out.
We have chosen to use the synthesizable subset of any popular HDL, to model
various sub-computations of various overall PG-based computations, for whom
we intend to automatically design (various) folded architectures.
Practically, the \textit{custom} design flow for this design methodology must
hand over at some point, RTL models to e.g. some \textit{standard} ASIC/FPGA design
flow. A \textit{case study} of successfully using this
design flow for prototyping a VLSI system is described in section
\ref{exp_sec}. The section also presents some details about the \textit{C++
tool} that has been implemented, to realize this methodology.

In this paper, we begin by giving a \textit{brief} introduction to
Projective Spaces in section \ref{sec2}, which is easy to grasp. It is
followed by a model of the nature of computations covered, and how they can
be mapped to PG based graphs, in section \ref{comp_model_sec}.  Section
\ref{fold_conc_sec} introduces the concept of folding for this model of
computation. The basic constructs for optimal scheduling, perfect access
patterns and sequences are introduced in section \ref{karm_sec}. Section
\ref{bipart_fold_sec} sketches out what kind of folding is desired from
regular bipartite graphs, while section \ref{outline_sec} brings out how
PG-based balanced regular bipartite graphs can be folded so,
\textit{optimally}. The details of various aspects of the design
methodology are brought out in section \ref{methodology_sec} next.
Especially, section \ref{aux_issue_sec} covers the detailed design problems
that are enlisted in section \ref{issues_sec}.  A scheme for pipelining the
folded designs to recover back some throughput, that is lost due to
trade-off, is covered in sections \ref{pipeline_sec}.  In section
\ref{mod_ref_sec}, we bring out the practical way of using this
methodology.  A note on addressing scalability concern in our design is
provided in section \ref{scale_sec}.  We provide specifications of some
real applications that were built using this methodology, in the
experiments section\,(section \ref{exp_sec}), before concluding the paper.

\section{Projective Spaces}
\label{sec2}

Projective spaces and their lattices are built using vector subspaces of
the \textbf{bijectively} corresponding vector space, one dimension high,
and their subsumption relations. Vector spaces being extension fields,
Galois fields are used to practically construct projective spaces
\cite{expanders}. However, throughout this work, we are mainly concerned
with subgraphs arising out lattice representation of Projective spaces,
which we discuss now. An overview of generating projective spaces from
finite fields can be found in \ref{appA}.

\begin{figure}[h]
\begin{center}
\includegraphics[scale=0.5]{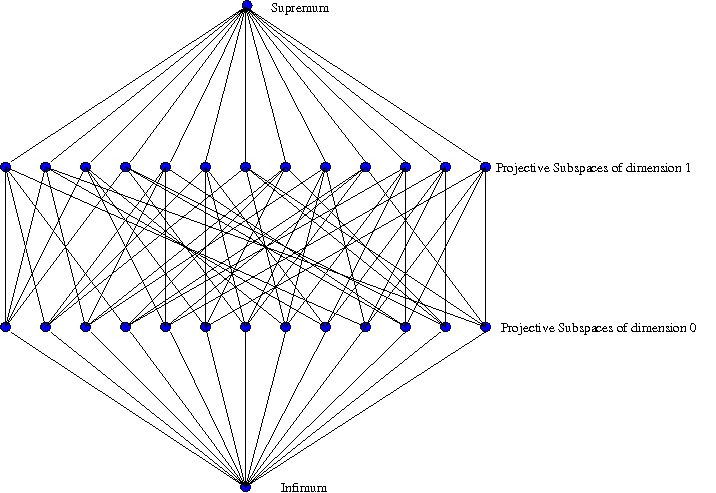}
\end{center}
\caption{A Lattice Representation for 2-dimensional Projective Space,
$\mathbb{P}(\mathbf{2},\mathbb{GF}(\mathbf{3}))$}
\label{pg_lat}
\end{figure}

It is a well-known fact that the lattice of subspaces in any projective
space is a \textbf{modular, geometric lattice} \cite{dam_pap}. A projective
space of dimension {\normalsize \textbf{2}} is shown in figure
\ref{pg_lat}. In such figure, the top-most node represents the
\textit{supremum}, which is a projective space of dimension {\normalsize
\textbf{m}} over Galois Field of size \textbf{q}, in a lattice for
{\normalsize $\mathbb{P}(\mathbf{m},\mathbb{GF}(\mathbf{q}))$}. The
bottom-most node represents the \textit{infimum}, which is a projective
space of (notational) dimension {\normalsize \textbf{-1}}. Each node in the
lattice as such is a projective subspace, called a \textbf{flat}. Each
horizontal level of flats represents a collection of all projective
subspaces of {\normalsize $\mathbb{P}(\mathbf{m},\mathbb{GF}(\mathbf{q}))$}
of a particular dimension. For example, the first level of flats above
infimum are flats of dimension {\normalsize \textbf{0}}, the next level are
flats of dimension {\normalsize \textbf{1}}, and so on. Some levels have
special names. The flats of dimension {\normalsize \textbf{0}} are called
\textit{points}, flats of dimension {\normalsize \textbf{1}} are called
\textit{lines}, flats of dimension {\normalsize \textbf{2}} are called
planes, and flats of dimension {\normalsize (\textbf{m-1})} in an overall
projective space {\normalsize
$\mathbb{P}(\mathbf{m},\mathbb{GF}(\mathbf{q}))$} are called
\textit{hyperplanes}. Many PG-based applications have models that are based
on two levels in this diagram, and connections based on their
\textit{inter-reachability} in the lattice. Out of these, the
\textit{balanced regular bipartite graphs made out of levels of points and
hyperplanes} have been used more often, because usually the applications
require the graph to have a high node degree, which this graph provides.

\subsection{Circulant Balanced Bipartite Graph}
A circulant balanced bipartite graph is a graph of {\normalsize \textbf{n}}
graph vertices on each side, in which the {\normalsize $\mathbf{i}^{th}$}
graph vertex of each side is adjacent to the {\normalsize
$\mathbf{(i+j)}\text{(modulo-n)}^{th}$} graph vertices of other side, for
each {\normalsize \textbf{j}} in a list \textbf{L} of vertex indices from
other side. A point-hyperplane incidence bipartite graph made from PG
lattice is a circulant graph; see Fig. \ref{pg_bbg}. We will be exploiting the circulance property
of PG bipartite graphs in our folding scheme.

\begin{figure}[h]
\begin{center}
\includegraphics[scale=0.7]{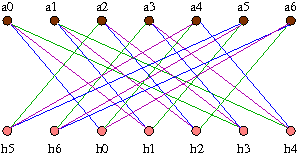}
\end{center}
\caption{An Example PG Circulant Bipartite Graph}
\label{pg_bbg}
\end{figure}

As will become clear from the constructive proof of main theorem \ref{th1},
this scheme can be extended to cover design of any system, whose DFG
exhibits a bipartite circulant nature, of any order.  However, a
\textit{practical} design methodology must target design of real systems.
Hence we stick to PG-based applications as our target real application area
of this design methodology.

\section{A Model for Computations Involved}
\label{comp_model_sec}
As mentioned earlier, we will be using a PG bipartite graph made from
points and hyperplanes in a PG lattice. In such graph, each point,
as well as hyperplane, is mapped to a unique vertex in the graph. Further,
a point and a hyperplane are
incident on one-another in this bipartite graph, if they are reachable via
some path in the corresponding lattice diagram. We state without proof,
that such bipartite graph is both balanced\,(both sides have same number of
nodes) and regular\,(each node on one side of graph has same
degree). For the proof, see \cite{hrishi_thesis}.

The computations that can be covered using this design scheme are mostly
applicable to the popular class of \textit{iterative} decoding algorithms
for error correcting codes, like Low-density Parity-check (LDPC) \cite{ldpc_survey}, polar
\cite{polar_pap} or expander codes \cite{sipspiel}. A representation
of such computation is generally available as a bipartite graph, though
it may go by some other \textit{domain-specific
name} such as \textit{Tanner Graph}. The nodes on each side of the bipartite graph
represent sub-computations (sequential circuits), which do not have any precedence orders. Hence
they can all be made to execute computations parallely. The
edges represent the data that is exchanged between nodes performing
sub-computations. Also, the nature of computation algorithm being
considered is such
that nodes on one side of the graph compute first, then nodes on the other
side of the graph. If the computation is iterative, then the computation
schedule so far is just repeated again and again. Such a schedule is
popularly known as \textit{flooding schedule}, since all nodes of
one side simultaneously send out data to nodes on other side. A bipartite graph is
undirected, and hence for \textit{visualization} as a \textit{Data Flow
Graph}\,(DFG), each of its 
edge can be replaced with two opposite-directed edges. Such an expansion is depicted in
figure \ref{dfg_fig}. Such a
refinement of problem model is \textit{only for conceptual clarity}, and not implemented in the
corresponding design flow. 
        Such a DFG may model both SIMD as well as MIMD
systems. Since we target design of PG-based applications using
this methodology, we assume throughout the remaining text that
\begin{enumerate}
\item The nature of parallel computation is SIMD.
\item The computation function realized by any node, is any computation
that can be realized using the a particular \textit{synthesis
subset} of various HDLs, described in section \ref{non_simd_pap_sec}.
\end{enumerate}

\begin{figure}[h]
\begin{center}
\includegraphics[scale=0.6]{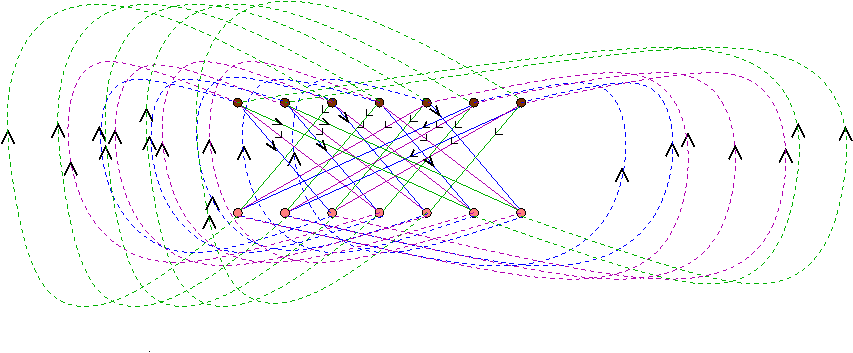}
\end{center}
\caption{A Visualization of Bipartite Graph as a Data Flow Graph}
\label{dfg_fig}
\end{figure}

Relaxing these assumptions leads to a tradeoff between optimality of system
performance, and ease of system implementation. Details of this tradeoff
can be found in section \ref{non_simd_pap_sec}.

After finishing the computations, nodes on any one side of the bipartite
graph transfer the resultant data for consumption of nodes on other side of
the graph, \textit{via} distributed storage in memory units. Usage of
\textbf{distributed memory} is common and fundamental requirement to folding the
graph using this method. Thus, one LMU per
node, just before its input along the data flow direction, is
the \textit{minimum} requirement for storing data which is
transferred within a bipartite graph\footnote{At times, to implement interconnect
pipelining to reduce signal delays in practical physical design of such
systems, memory elements may also be present at the output.}.
An easy way of implementing distributed memory on both sides is to collocate
local/on-chip memory of each physical node with each required PMU.

\section{Conflict-free Communications Primitives for PG Graphs}
\label{karm_sec}
The scheduling model used in the folding scheme is based on
\textit{Karmarkar's template}\cite{karm1}. PG lattices possess structural
regularity in form of \textit{circulance}, and this property has been
exploited in scheduling of general parallel systems. Karmarkar was able to
come up with a parallel communication method to realize various ``nice
properties'' in scheduling, which are enlisted later in the section. He
discovered two \uline{memory-conflict free communication primitives} using
bipartite graphs derived from 2-dimensional Projective Space Lattices
\cite{karm1}.

\begin{figure}[!h]
\centering
\includegraphics[scale=0.7]{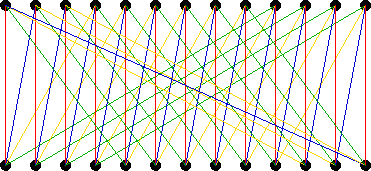}
\caption{Perfect Access Primitives in a PG balanced bipartite graph}
\label{pg_tan}
\end{figure}

Let {\large \textbf{n}} processing units be placed in place denoted by the
lines, and {\large \textbf{n}} memory units placed in place denoted by the
points, in a PG bipartite graph. Consider a binary operation that is to be
scheduled on these processing units in SIMD fashion. Let it take two
operands as inputs (reads from two memory locations) per cycle, and modify
one of them (writes back in one memory location) as output. The binary
operation is preferred since the required memory
unit is then a \textbf{dual-port} memory, something that is easily
commercially-off-the-shelf (COTS) available. The schedule of
memory accesses for a collection of such operations, that corresponds to a
\textbf{particular} complete set of line-point index-pairs, for
simultaneous parallel execution over one cycle on all processing units is
known as a \textbf{Perfect Access Pattern}. A set of such patterns, with
some application-defined order imposed on them, is known as a
\textbf{Perfect Access Sequence}. Such \textbf{particular} complete set of
line-point index pairs is generated by exploiting circulant nature of PG
bipartite graph. On \textit{each node} on \textbf{one} side of the graph, two edges are
chosen such that they are shift-replicas of the two edges chosen for its
neighboring node.  For example, in figure \ref{pg_tan}, the set of 13 red
and 13 green edges forms one \textit{Perfect Access Pattern}, and 13 yellow
and 13 blue edges another \textit{Perfect Pattern}. These two perfect
patterns (like these two), when sequenced in \textit{arbitrary} order,
form a \textit{Perfect Access Sequence}.  The properties of such an
execution of processing
unit-memory unit communication are as follows \cite{karm1}.
\begin{enumerate}
\item There are no read or write conflicts in memory accesses.
\item There is no conflict or wait in processor usage.
\item All processors are fully utilized.
\item Memory bandwidth is fully utilized.
\end{enumerate}

\subsection{Generalization}
The cost of a perfect access sequence is \textit{{\large
$\mathbf{\gamma/2}$} cycles}, where {\large $\mathbb{\gamma}$} is
the degree of each node in bipartite graph. There
can be possibly \uline{alternative communication primitives}, which can
have different communication costs over the same projective plane.
Generalizing
beyond binary operation scheduling to \textbf{n}-ary operation scheduling on
computing nodes reduces the communication cost, but leads to complexity
of the memory unit controller's design/area/power.

Practically, there are many parallel computational
problems, implementable in hardware, whose communication graph has been
derived out of higher-dimensional projective spaces. Two such problems,
that were worked out by us, are LU decomposition\,(exploiting a
\textbf{4}-dimensional underlying projective space) \cite{mat_pap}, and the
DVD-R decoder\,(exploiting a \textbf{5}-dimensional underlying projective
space) \cite{expanders}. In \cite{fold2_techrep}, it is proven in detail
that Karmarkar's scheme of decomposing a projective plane into perfect
access patterns can indeed be extended to point-hyperplane graphs of
\textbf{arbitrary} dimensional Projective Space. For sake of brevity, the
proof is not repeated here.

\subsection{Suitability of Perfect Access Patterns for Other Computations}
\label{non_simd_pap_sec}

We explain now that \textbf{any}
\textit{synthesizable} \textbf{sequential logic} can represent the computation
meant by the `single instruction' in SIMD model, as long as in its
multi-input Mealy machine representation,
each transition is governed by arrival of a particular input
signal, and \textbf{not on the value} of the
signal. Thus, in a given state, we assume that such FSM, in a
given state, accepts a compatible signal arrival event, transitions into a
\textit{unique} state, and optionally outputs a \textit{unique} set of signals,
\textit{irrespective} of the value of the input signal. In our computation model, each
input edge incident on a vertex is treated as a signal. Multiple inputs can
arrive simultaneously in sequential logic, in which case the event is a
compound signal event. Since we use SIMD
model, the labeling of edges of all vertices on one side of bipartite
graph, to represent signals, can be made \textbf{isomorphic} easily. Such
labeling allows FSMs of all the node computations to move in \textbf{synchronized
fashion}, requiring \textit{inputs in same sequence} on all
nodes on one side of bipartite graph. This is because FSM
model of any sequential logic computation imposes a legal order requirement on its inputs
, in order to reach its end state. Further, the legally ordered set of such inputs
required by the `single instruction' may not cover the complete set of
possible inputs (edges) on each node. As long as \textbf{same} subset
of inputs, in \textbf{same} sequence, is needed by each node to reach their end
states, the collection of such \textit{subsequences} can be used as a
\uline{perfect access sequence} required by the
computation of `single instruction'. These
subsequences must be synchronized at each clock cycle, for load balancing;
there cannot be \textit{gaps} in their scheduling.  We can then
break such common sequence into perfect access patterns, and use
the basic result
of folding a perfect access pattern (see theorem \ref{th1}) to optimally schedule
each such computations. Because we have the choice of picking up order
while forming a \textit{perfect access sequence} from the set of
\textit{perfect access patterns} (see section \ref{karm_sec}), we also have
a choice in scheduling and ordering the input arrivals. Thus, we can always
force the same order, as required by the sequential logic, on the perfect
access `sequence'.
A combinational logic computation is treated as a special case of
sequential logic computation.

The application classes that we realistically target (described in section
\ref{intro_sec}) have computations (e.g. accumulation operator), that
naturally obey the restriction described
above. Their multi-input Mealy machine
model is a set of disjoint equal-length paths, between unique
\textit{start} and \textit{end} states. The length of each path is
$\gamma$, i.e. each legal input sequence to the state machine requires signals
on \textbf{all} edges to arrive, in some \textit{permutation order}, before
completion of computation. The number of such paths in these models is
equal to $\gamma!$, though in our \textit{generalized} model, it can be
$\leq$ $\gamma!$.

Going further, \textit{suppose} we \textbf{relax} the SIMD
assumption, and assume MIMD model of computation for the system under
design. In such a case, there will be no restriction whatsoever on the
sub-computation that is happening on each node in a particular
cycle, and their computation times. The
computations may be different, e.g.  addition and subtraction. As long as
all nodes on one side of the graph operate on the same number of operands
at a time, and take same number of cycles to complete, the foldability of
graph derived in this report will remain applicable. One may further relax
the same computation time constraint on these sub-computations, by
implementing a \textit{barrier synchronization} on either side of the data
flow graph. All such relaxations need to be \textit{annotated/added} to the
system model (Tanner Graph), and hence form the \uline{\textbf{first}
level of refinement} (specification refinement) of the DFG, which is an
\textbf{optional level}. It is straightforward to notice that while
applying this design methodology to MIMD systems retains the ease of
engineering the system, as in SIMD case, there are chances that the system
may lose some amount of performance optimality (e.g. due to mandatory
barrier synchronization).

\section{The Concept of Bipartite Graph Folding}
\label{fold_conc_sec}
Semi-parallel, or folded architectures are hardware-sharing architectures,
in which hardware components are \uline{shared/overlaid} for performing
different parts of computation within a (single) computation. As such,
folded architectures are a result of fundamental space-time 
\textit{complexity} trade-off in (parallel) computation. This in turn
manifests in form of area-throughput trade-off during its (parallel)
implementation.

In its basic
form, folding is a technique in which more than one algorithmic operations
of the same type are mapped to the same hardware operator. This is achieved
by \uline{time-multiplexing} these multiple algorithm operations of the
same type, onto \textit{single} computational unit at system run-time.
Hereafter, we define \uline{\textbf{logical} processing unit(\textbf{LPU})} as the
\textit{logical} computational unit associated with each node of the graph,
while \uline{\textbf{physical} processing unit(\textbf{PPU})} as the
\textit{physical} computational unit associated with each node of the graph.
Multiple LPUs get \textit{overlaid} on single PPU, after
folding. We also define the equivalent term for \textit{overlaid}
memory unit as \uline{\textbf{physical} memory unit(\textbf{PMU})}, which
is an overlay of multiple \uline{\textbf{logical} memory units(\textbf{LMUs})}.

\begin{figure}[h]
\begin{center}
\includegraphics[scale=0.4]{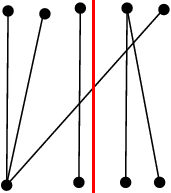}
\end{center}
\caption{(Unevenly) Partitioned Bipartite DFG}
\label{fold_bp}
\end{figure}

The balanced bipartite PG graphs of various target applications perform
parallel computation, as described in section \ref{comp_model_sec}.  In its
\textbf{classical} sense, a folded architecture represents a partition, or
a \textit{collection of} \textbf{folds}, of such a (balanced) bipartite graph\,(see figure
\ref{fold_bp}). The blocks of the partition, or folds can themselves be
\textit{balanced} or \textit{unbalanced}; partitioning with unbalanced block sizes entails no
obvious advantage.  The \textit{computational} folding can be implemented
after (balanced) graph partitioning in \textbf{two dual ways}. In
the \textbf{first} way, that is
used in \cite{cacs_pap}, \cite{dam_pap}, the \uline{within-fold}
computation is done
\textit{sequentially}, and \uline{across-fold} computation is done
\textit{parallely}. Such a scheme is generally
called a \textit{supernode-based folded design}, since a \textit{logical}
supernode is held responsible for operating over a fold. \textbf{Dually},
the \uline{across-fold} computation can be made \textit{sequential} by scheduling
first node of first fold, first node of second fold, $\ldots$
\textit{sequentially} on a single module.  The \uline{within-fold}
computations, held by various nodes in the fold, can hence be made
\textit{parallel} by scheduling them over different hardware modules. This
scheme is what we cover in this paper.
Either way, such a folding is represented by a time-schedule, called the
\textbf{folding schedule}. The schedule tells that in each machine cycle,
which all computations are parallely scheduled on various
PPUs,
and also the \textit{sequence} of clusters of such parallel computations
across machine cycles.

\subsection{Folding PG-based Bipartite Graphs}
\label{bipart_fold_sec}
We first sketch out, how a PG bipartite graph can be folded. Generally
folding is performed by partitioning the vertex sets of the bipartite
graph, and \textit{overlaying} them on various available PPUs.
As such, general folding schemes \textit{are not} able to overlay the edge sets onto
each other. It potentially results in reconfiguring the
interconnection between physical units at run-time, whenever a new
fold has to be scheduled. What \textbf{stands
out} in case of using our folding scheme is that edges also
get overlaid. Hence the \textit{entire run-time overhead} of reconfiguring the
interconnect via various mux selections \textit{is saved}.

In a PG balanced bipartite graph made from points and hyperplanes of
{\normalsize $\mathbf{n}$}-dimensional projective space over {\normalsize $\mathbb{GF}(p^s)$},
{\normalsize $\mathbf{P}(n,\mathbb{GF}(p^s))$}, the number of nodes on either side is
{\large $\mathbf{J}$} = {\large $\mathbf{\frac{p^{s (n+1)}-1}{p^s-1}}$}, while the degree of
each node is {\large $\mathbf{\gamma}$} = {\large $\mathbf{\frac{p^{s
n}-1}{p^s-1}}$}. Here,
{\normalsize \textbf{p}} is any prime number, while {\normalsize \textbf{s}} is any
natural number.
For vertex partitioning, as discussed earlier, we choose to have e.g. $\mathbf{1}^{st}$ PPU performing $\mathbf{1}^{st}$ left node computation in a cycle, then $\mathbf{5}^{th}$
left node computation in next cycle, and so on. By doing so, it so happens,
as we prove later, that the destination vertex of each edge incident on
various nodes across various partitions of one side of the graph,
that are mapped to same PPU post folding, remains
\textbf{identical}. Due to dual-port memory unit restriction, the
computation by each PPU can only be performed
across multiple cycles\,(\textbf{2} inputs possible per cycle). Hence we also
need to partition the edge set of each node, generally into subsets of \textbf{2}
edges, as depicted in figure \ref{pg_tan}.

By applying perfect access patterns and sequences \cite{karm1} for
inter-unit communication, that are applicable for \textit{all possible} 
point-hyperplane bipartite graphs, the \textit{overlaid} edge partitioning mentioned above can be
\textit{readily achieved}.  Recall that a perfect access pattern stimulates only a
fraction of edges per node in a cycle. Hence we focus our efforts on
evolving the vertex partitioning  only.
For practical designs, to avoid $>$ \textbf{2} concurrent accesses to a
memory unit in a machine cycle, we assume that edge-partitioning has
already been done\,(forming perfect access sequence), and that we are trying
to do a vertex partitioning \textbf{over} each Perfect Access Pattern
within the sequence. Further, in vertex partitioning, as
reasoned earlier, we focus on creating
\textit{balanced, equal-factor} partitions only; refer figure \ref{fold_bp}. However, the
methodology can be extended easily to handle unequal-factor folding of both
sides as well.

\section{Core Folding Scheme}
\label{outline_sec}
In the subsequent text, we assume that associated with each node or
PPU, there is one (distributed) PMU,
using which data can be transferred across the bipartite graph for computation. We have
already mentioned this assumption before, in section \ref{comp_model_sec}.
To recall from \ref{intro_sec}, \uline{\textbf{logical}
processing unit (\textbf{LPU})} is defined as the \textit{logical}
computational unit associated with each node of the graph, while
\uline{\textbf{physical} processing unit (\textbf{PPU})} as the
\textit{physical} computational unit associated with each node of the graph.
The equivalent term for \textit{overlaid} memory unit is
\uline{\textbf{physical} memory unit (\textbf{PMU})}, which is an overlay of
multiple \uline{\textbf{logical} memory units (\textbf{LMUs})}.
Hence in the initial architecture, there are \textbf{J}
LPUs and LMUs of one type, and
another \textbf{J} LPUs and LMUs of another type.
This architecture represents the \uline{\textbf{second} level of
refinement}\footnote{first mandatory, to-be-implemented level of
refinement} of the data
flow graph, and is more detailed in section \ref{sys_arch}. As per
the model of computations to be scheduled on this architecture (section
\ref{comp_model_sec}), LPUs of one type need to read their input data from LMUs of
the other type. The
\underline{core problem} that we tackle first is to prove that using an
equal number of LPUs and LMUs, where the number is
\textbf{any} \textit{factor of} \textbf{J}, and interconnecting them in
\textit{specific way}, the
\textit{necessary} data flow between them in an unfolded PG bipartite graph based
computation can still be achieved \textit{optimally}. We build the design methodology around
this \textbf{main result}.

\subsection{Problem Formulation}
\label{prob_form_sec}

Suppose we fold both sets of nodes by a factor of {\normalsize \textbf{q}} in a PG
balanced bipartite graph.  Hence there are \textbf{J/q} PPUs
and PMUs
of either type. Since overall number of edges in the
non-folded regular bipartite graph is $\mathbf{\gamma\times J}$($\gamma$
defined in section \ref{bipart_fold_sec}), the required size
of each PMU to store all data
corresponding to these
many edges is $\mathbf{q\times\gamma}$. Our unit of computation is a
fold of one row of nodes, each of which has $\gamma$
inputs/outputs. If this fold were to impose \textit{uniform
load/storage requirements} on each of the \textbf{J/q} memories, then
the uniform (storage/communication) load imposed by outputs of \textbf{J/q}
PPUs on \textbf{J/q} PMUs is trivially $\gamma$.

Given that we have \textbf{J/q} PPUs and PMUs physically available, one
question is whether it is possible to generate perfect
\textbf{patterns} using \textbf{J/q} elements of either type (PPUs or PMUs). If this were true, then it will lead to uniform
load ($\gamma$) on the \textbf{J/q} PMUs, since we know that perfect access patterns
impose balanced loads \cite{karm1}. Combining such patterns
will give a perfect access sequence. We discuss some possible
approaches to this question now.

To have a embedded perfect access pattern, one option is that \textbf{J/q} nodes of both
types, and their interconnection becomes a embedded PG sub-geometry in
itself. For that, 
\textbf{J/q} must take a value of
form {\large $\mathbf{\frac{p_{1}^{s_{1} (n+1)} - 1}{p_{1}^{s_{1}} - 1}}$} for some
prime {\normalsize $\mathbf{p}_1$} and non-negative integer {\normalsize $\mathbf{s}_1$}. This
is the cardinality of the set of hyperplanes in some
{\large $\mathbf{P}(n,\mathbb{GF}({p_1}^{s_1}))$}. In such
a case, we would need to study such structure-ability of \textbf{J} for
various values of \textbf{p} (its base prime) and \textbf{q} (its desired
factors).

If this were possible, node connectivity of such embedded geometry, from first principles,
will be {\large $\mathbf{\frac{p_{1}^{s_{1} n} - 1}{p_{1}^{s_{1}} - 1}}$}
\cite{karm1}. However, each node needs all of {\large
$\gamma\,=\,\mathbf{\frac{p^{s n} - 1}{p^{s} - 1}}$}
inputs, where $p^s$ is order of the base Galois field of {\normalsize
\textbf{n}}-dimensional projective space
under consideration, for otherwise, their computation will be incomplete.

As an example, let $\mathbf{p}$ = 3 and $\mathbf{s}$ = 2. Then $\mathbf{J}$
= 91 and {\normalsize $\gamma$} = 10.
Now $\mathbf{q}$ = 7 is a factor of 91. If we fold each row of node 7 times, then
\textbf{J/q} = 13. An order-13 regular bipartite graph is possible when
$\mathbf{p_1}$ = 3 and $\mathbf{s_1}$ = 1. However, by definition, such a smaller graph has
its regular node degree $\gamma^{'}$ = 4, while we need it to be 10 itself.

The solution lies in simply increasing the
LMU size and number of accesses per
LMU.  As one can see, in general for projective spaces over
\textit{non-binary} Galois Fields, $\gamma$ is divisible by 2. When we take
2-access at a time, we can form a perfect access pattern in the \textbf{J/q}-sized
fold of a regular bipartite graph as detailed in theorem \ref{th1}, for
\textbf{ANY} \textbf{q}. We
later easily extend the same pattern generation for graphs derived from
projective spaces based on \textit{binary} Galois Fields.

\subsection{Folding by \textbf{ANY} Factor}
\label{any_fold_sec}

We now generalize our earlier analysis suitably and make the \textit{final}
statement.
\begin{thm}
\label{th1}
It is possible to generate a (folded) perfect access
pattern, from a non-folded perfect access pattern, using \textbf{J/q}
LPUs and LMUs of a fold that belongs to the bipartite
graph based on $\mathbf{P}(n,\mathbb{GF}(p^s))$, for \textbf{ANY}
\textbf{q} that divides \textbf{J}.
\end{thm}

\begin{figure}[h]
\begin{center}
\includegraphics[scale=0.9]{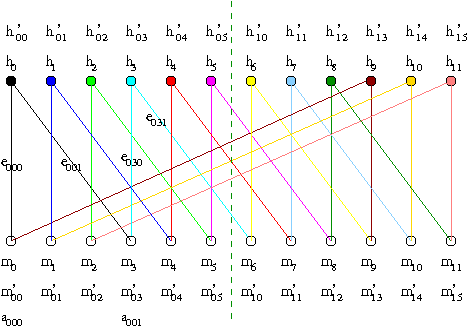}
\end{center}
\caption{Example Circulant Representation of PG Bipartite Graph}
\label{perf_pat}
\end{figure}

\begin{proof}
The two important properties used in this proof are properties of
\textit{modulo addition}, and
circulance of PG-based balanced bipartite graph. As
mentioned earlier, PG-based bipartite graph is a circulant
graph.

For all notations as well as all
\textit{representative} indices that we use hereafter in the paper, we
follow figure \ref{perf_pat}. Let
the unfolded set of computations (hyperplanes) be represented as $\{h_i: 0
\leq i \leq \mathbf{J}\}$. After folding, let the new set of LPUs be represented as $\{h^{'}_{ji}: 0 \leq j \leq \mathbf{q},\, 0 \leq i
\leq \mathbf{J/q}\}$. Similarly, let the unfolded set of storages (points)
be represented as $\{m_i: 0 \leq i \leq \mathbf{J}\}$. After folding, let
the new set of dual-port LMUs be represented as $\{m^{'}_{ji}: 0
\leq j \leq \mathbf{q},\, 0 \leq i \leq \mathbf{J/q}\}$.
Given a \textit{subgraph} which corresponds to
any one \textbf{full} (non-folded) perfect
pattern which has to be vertex-folded, let \textit{some} two edges of
some node marked by $h^{'}_{ji}$ be $e_{ji0}$ and $e_{ji1}$.

\textit{Overall}, $h^{'}_{00}$ being the first node in the $0^{th}$ fold,
assume that it is connected via
$\{e_{000}$, $e_{001}$, $\cdots$, $e_{00(\gamma -1)}\}$ edges to
different LMUs, where $\gamma$ =
{\large $\mathbf{\frac{p^{s n} - 1}{p^{s} - 1}}$}. Let us assume that the regular
bipartite graph has been \textbf{re-labeled} and \textbf{re-arranged}, such that
circulance is in as explicit form as shown in figure
\ref{perf_pat}.
Using \textit{circulance property} of a point/hyperplane in such graph
results in mapping of that point/hyperplane, and all its edges, to one of
its immediate neighbor node on the same side. Let us
denote the \textbf{ends} of first two edges from
hyperplane $h^{'}_{00}$, $a_{000}$ and $a_{001}$.
Without loss of generality, assume hyperplanes
represent the set of computations being done currently, while points
represent the set of LMUs from which input/output to computations
is happening. Indices $a_{000}$ and $a_{001}$ belong to interval
[\textbf{0, J}], and need to be re-mapped to index set of physically available
LMUs, [\textbf{0, J/q-1}]. For this, we take
\textit{remainder} modulo-(\textbf{J/q}) of $a_{000}$ and $a_{001}$,
and denote the new indices by $a^{'}_{000}$ and $a^{'}_{001}$. The two new indices are either equal
or they are not equal. In either case, when we re-index ends of
the two edges of any hyperplane $h^{'}_{0i}$, from points $a_{0i0}$ and
$a_{0i1}$ to points $a^{'}_{0i0}$ and $a^{'}_{0i1}$, then by circulance property, the \textbf{shift} between
$a^{'}_{0i0}$ and $a^{'}_{000}$ (or between $a^{'}_{0i1}$ and $a^{'}_{001}$)
is \textbf{equal to} the shift between $h^{'}_{0i}$ and
$h^{'}_{00}$. After such successive re-indexing \textbf{J/q} times,

\begin{enumerate}
\item The set of hyperplane indices used covers up all the values between \textbf{0} and
$\mathbf{\left(\frac{J}{q}-1\right)}$.
\item By virtue of modulo-$\mathbf{\left(\frac{J}{q}\right)}$ addition by \textbf{1},
$\mathbf{\frac{J}{q}}$ times, the set
of new first point indices covers all the values between \textbf{0} and
$\mathbf{\left(\frac{J}{q}-1\right)}$. Similarly, the set of new second point indices covers all the
values between \textbf{0} and $\mathbf{\left(\frac{J}{q}-1\right)}$ as well.
\end{enumerate}
It is straightforward to check that all \textbf{necessary and sufficient
conditions} for generation of perfect access patterns and sequences
\cite{karm1} get immediately satisfied. Hence we have constructively proven
that such folded perfect access patterns exist for PG bipartite graphs, which by
definition, impose \textbf{perfectly balanced} (communication)
\textit{load} on various modules such as PMUs and PPUs. For certain error-correction computations, especially
such memory efficiency is highly desirable \cite{tarable}.
\end{proof}
\begin{corl}
\label{corl1}
As an important corollary, it is easy to prove that the total number of
PMUs accessed by each PPU, $\rho$, is $\leq$
$\gamma$, as well as $\leq$ \textbf{J/q}.
\end{corl}

We now also prove one of our earlier claims: that edges get overlaid while
folding a PG-based bipartite graph for \textbf{ANY} factor \textbf{q}.

\begin{thm}
\label{th2}
It is possible to provide a complete one-to-one mapping of between two sets
of edges, belonging to any two folds of a PG bipartite graph, created using \textbf{ANY}
\textbf{q} that divides \textbf{J}. Each edge set of a fold is defined as
the set of all edges that are incident on \textit{any one side}
of nodes of that fold.
\end{thm}

\begin{proof}
Let us consider any two fold indices \textbf{x} and \textbf{y} to prove
overlaying of edges. For each edge $e_{xjk}$, the $\mathbf{k^{th}}$ edge
incident on $\mathbf{j^{th}}$ node of $\mathbf{x^{th}}$ fold, consider
$e_{yjk}$, again $\mathbf{k^{th}}$ edge incident on $\mathbf{j^{th}}$ node
of \textit{different} fold, $\mathbf{y}$. These edges are shift-replicas of
each other in the \textit{unfolded} graph. Let the remote end point of
$e_{xjk}$ is $a_{xjk}$, and that of $e_{yjk}$ be $a_{yjk}$ in the
\textit{unfolded} graph. Then, by virtue of circulance, the remote end
point post-folding of $e_{xjk}$ will be
$\left(a_{xjk}\right)\left(\text{mod-}\mathbf{\frac{J}{q}}\right)$, and
that of $e_{yjk}$ must be
$\left(a_{yjk}\right)\left(\text{mod-}\mathbf{\frac{J}{q}}\right)$ $=$
$\left(a_{xjk} + \left|x-y\right|\cdot\mathbf{\frac{J}{q}}\right)
\left(\text{mod-}\mathbf{\frac{J}{q}}\right)$. This can be simplified to
$\left(a_{yjk}\right)\left(\text{mod-}\mathbf{\frac{J}{q}}\right)$, thus
proving that
$\left(a_{yjk}\right)\left(\text{mod-}\mathbf{\frac{J}{q}}\right)$ $=$
$\left(a_{xjk}\right)\left(\text{mod-}\mathbf{\frac{J}{q}}\right)$ for any
choice of \textbf{x} and \textbf{y}. Since all the $\mathbf{j^{th}}$ nodes
of all folds overlay on each other anyway, such edges which are incident on
these nodes, and also have identical end points post folding, will surely
coincide.
\end{proof}

The above edge overlay is a \textbf{significant property} of this
folding scheme, since it is a \uline{perfect overlay}. That is,
\textbf{each} edge incident on some node of a particular fold,
\uline{uniquely} overlays on some edge of an overlaid node of any other
fold. This advantage simplifies the system design by totally eliminating the
use of switches for connection reconfigurations.

\subsection{Lesser Memory Units}
For some values of $\mathbf{q}$, it is possible that \textbf{J/q} becomes
less than $\gamma$, the degree of each node. This implies that the number
of inputs/outputs per PPU is greater than the number of PMUs. It is straightforward to see the our folding scheme still satisfies
all the prerequisite axioms for generation of perfect access patterns and
sequences, and hence is valid for this case as well.

\section{A Design Methodology Using the Folding Scheme}
\label{methodology_sec}
In this section, we provide a set of algorithms for designing various aspects of intended system,
including memory layout/sizing, communication subsystem design etc., of a folded PG architecture.
This corresponds to \textit{remaining level of refinements}, of the
system model. The output at the end of these refinements is
expected to be the \uline{RTL specification} of the overall
system, which includes cycle-accurate behavior of each component.
Beyond the last level, standard RTL synthesis tools can be integrated
into the design flow for the remaining refinement. This is possible, since
beyond RTL, standard design flows are available, and have to be practically
used. The last subsection summarizes the overall methodology (till RTL stage).

Throughout this chapter, unless stated otherwise, we will consider the PG bipartite graph
made from 3-dimensional projective $\mathbf{P}(3,\mathbb{GF}(2^2))$, as a
\uline{running example}. It has 15 nodes on either side (points
and hyperplanes), and each node is connected to 7 nodes on other side of
the graph. The hyperplane-point incidence is shown in table
\ref{h_p_ex_tab}. Each row of the table lists the points that are
incident on the correspondingly labeled hyperplane. The incidence relations
have been calculated by constructing the Galois extension field, as
outlined and exemplified in appendix \ref{appA}. A corresponding
bipartite graph is shown in Fig. \ref{15_pg_fig}.

\begin{table}[h]
\caption{Point-Hyperplane Correspondence in 3-d Projective Space over
$\mathbb{GF}${(2)}}
\label{h_p_ex_tab}
\centering
{\normalsize
\begin{tabular}[!h]{|c|c|}
\hline
\textbf{Hyperplane no.} & \textbf{List of Points} \\ \hline \hline
0 & \{0, 1, 2, 4, 5, 8, 10\} \\ \hline
1 & \{1, 2, 3, 5, 6, 9, 11\} \\ \hline
2 & \{2, 3, 4, 6, 7, 10, 12\} \\ \hline
3 & \{3, 4, 5, 7, 8, 11, 13\} \\ \hline
4 & \{4, 5, 6, 8, 9, 12, 14\} \\ \hline
5 & \{5, 6, 7, 9, 10, 13, 0\} \\ \hline
6 & \{6, 7, 8, 10, 11, 14, 1\} \\ \hline
7 & \{7, 8, 9, 11, 12, 0, 2\} \\ \hline
8 & \{8, 9, 10, 12, 13, 1, 3\} \\ \hline
9 & \{9, 10, 11, 13, 14, 2, 4\} \\ \hline
10 & \{ 10, 11, 12, 14, 0, 3, 5\} \\ \hline
11 & \{ 11, 12, 13, 0, 1, 4, 6\} \\ \hline
12 & \{ 12, 13, 14, 1, 2, 5, 7\} \\ \hline
13 & \{ 13, 14, 0, 2, 3, 6, 8\} \\ \hline
14 & \{ 14, 0, 1, 3, 4, 7, 9\} \\ \hline

\end{tabular}
}
\end{table}

\begin{figure}[h]
\begin{center}
\includegraphics[scale=0.5]{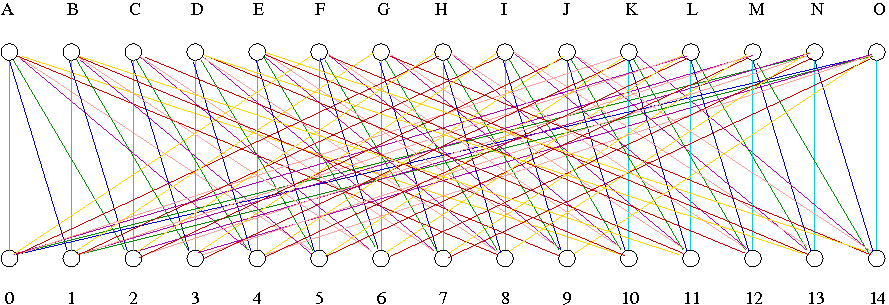}
\end{center}
\caption{Re-labeled (15,15) PG Bipartite Graph}
\label{15_pg_fig}
\end{figure}

To again recall from \ref{intro_sec}, \uline{\textbf{logical}
processing unit (\textbf{LPU})} is defined as the \textit{logical}
computational unit associated with each node of the graph, while
\uline{\textbf{physical} processing unit (\textbf{PPU})} as the
\textit{physical} computational unit associated with each node of the graph.
The equivalent term for \textit{overlaid} memory unit is
\uline{\textbf{physical} memory unit (\textbf{PMU})}, which is an overlay of
multiple \uline{\textbf{logical} memory units (\textbf{LMUs})}.

\subsection{System Architecture and Data Flow}
\label{sys_arch}
As discussed earlier in section \ref{comp_model_sec}, a PG bipartite graph
represents a data flow graph, with each side of the bipartite graph
representing multiple instances of one type of computation. These
two types of component computations happen one after the other in
\textit{flooding scheduling}.  To design such a system, we first
\textbf{refine} the PG bipartite graph into an architecture diagram at the
\uline{\textbf{second} level of refinement}.  At this \textit{computation}
refinement level, we turn the specification into a high-level architecture.
For this, first the value of fold factor, \textbf{q}, is chosen.
Recall that \textit{first level of refinement} is optional. Hence in such
architectural model, there
are two sets of \textbf{J/q} PPUs, and two sets of
\textbf{J/q} PMUs. One set of PMUs
is \textit{collocated} with one set of PPUs, and similarly the
remaining two. One-to-one mapped local channels are added between 2 ports
of each PPU, and the 2 ports of collocated PMU. Thus the
read/write access \textit{between} each $\langle$PPU, PMU$\rangle$ pair is \textit{local}.
Based on requirements imposed by the application, one set of
collocated $\langle$PPU, PMU$\rangle$ pair uniquely
corresponds to a subset of overlaid hyperplane nodes, and similarly the
other set of collocated $\langle$PPU, PMU$\rangle$
pair uniquely corresponds to a subset of overlaid point nodes.
Based on such roles, \textbf{two sets} of connections derived from
\textit{folded} PG bipartite graph, in form of channels, are added between
set of PMUs of one side, set of PPUs of the other side,
for both the sides. A folded architecture, which arises from
such \textit{second level refinement} of PG
bipartite graph, is depicted in figure \ref{fold_arch}.
This model qualifies to be a \textit{transaction-level model}, as
defined in \cite{gajski_tlm_pap}.

The model of each PPU \textit{after this refinement}
is an \textit{untimed} model that describes its internal computation
in some chosen model of computation, \textbf{after
modifications} that relate to overlaying of such units. This model
cannot be a cycle-accurate model, since
specification of that requires the knowledge of sequence in which inputs
arrive. This sequence is dependent on design option chosen as in section
\ref{pat_sec}, something that is part of next level of refinement.
Hence the cycle-level details of this modification are detailed later in
section \ref{micro_mod_sec}, as part of \textit{third level of
refinement}.
Similarly, the model of PMU \textit{after this refinement} is a
partially complete model, which includes a properly-sized RAM and a
placeholder for an address generation component. Details of this component
are filled at \textit{fourth level of refinement}, as per section \ref{add_gen_sec}.
The internal layout of these PMUs is described in section \ref{mem_lout_sec}.

For normal (non-folded) flooding scheduling of such computation, we assume
the \textit{convention} that
first set of PPUs read the required data from PMUs of
the \textit{other} side, utilizing the services of a PG interconnect. They
then write the output data in their \textit{local} PMU. For the
next half of computation, the second set of PPUs now
access the PMUs of the first type via the interconnect, to read in
their data (output by the first set of PPUs). They also write
back their output in their local PMUs, to be later read in by the
first set of PPUs in the next iteration.

\begin{figure}[h]
\begin{center}
\includegraphics[scale=0.72]{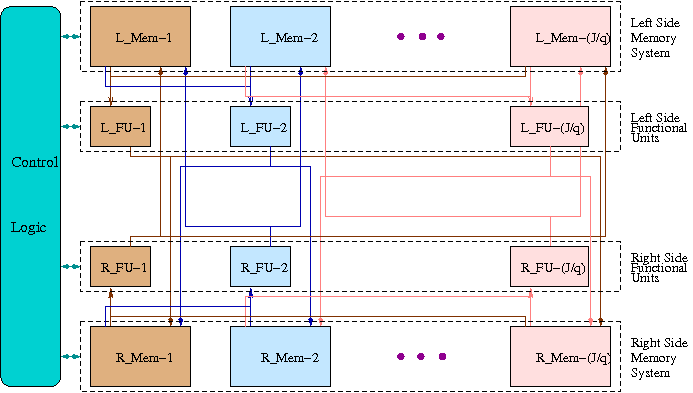}
\end{center}
\caption{High-level Architecture of Folded PG Bipartite Computing System}
\label{fold_arch}
\end{figure}

Such high-level system architecture next needs to be completed with details
of further componentization (e.g., separating address generation unit from
actual storage in PMU), thus taking it to last two refinement levels.
This folding design is explained over next few sections.

\subsubsection{Handling Prime Number of Computational Nodes}
\label{prime_sec}
For some values of \textbf{p} and \textbf{s}, the number of nodes on one
side of bipartite graph, $\mathbf{J}$ = {\large $\mathbf{\frac{p^{s
(n+1)}-1}{p^s-1}}$},
may be a prime number. For such number, no factor exists, based on which
\textit{second level of refinement} can be carried out. To still design for
folding, we proceed as follows. Since this step is not always
needed, a reader may skip this subsection in first reading. We \textit{add} a small
number of \uline{dummy nodes} to the graph towards one end of the graph,
on both sides. The number of additional nodes can be at least
one (in which case, the total number of nodes becomes an even
number). We then \uline{\textit{convert}} the
original circulant bipartite graph into a
\textbf{expanded} \textit{circulant} bipartite
graph, using algorithm \ref{prime_algo} described next. If the
new graph is not kept circulant, then scheduling \textit{across
folds} will entail changing of wiring at runtime, something that is
undesirable. This is because theorem \ref{th1} holds only for circulant
graphs. The remaining steps in the folding design, after this
optional expansion, remain \textbf{identical}.

In the following algorithm, if we add $\alpha$ dummy nodes to the graph,
then we also add at \textit{maximum} $\gamma$ \textit{dummy edges} per
\textit{retained} node. All the edges \textit{retained} from
earlier graph are called \uline{real edges}; and all that are newly added
as per algorithm will be called \uline{dummy edges} hereafter. The essence of the algorithm is to grow a union of $\gamma$
perfect matchings into a union of \textit{at maximum} ($2\cdot\gamma$)
perfect matchings as follows. A perfect access sequence is simply the
\textit{disjoint union} of various perfect matchings in a balanced bipartite
graph; see \cite{fold2_techrep}. Let nodes on one side of the original
graph be denoted as $h_0$, $h_1$, $\cdots$, $h_{\mathbf{J}-1}$, and nodes
on other side as $a_0$, $a_1$, $\cdots$, $a_{\mathbf{J}-1}$. By
abuse of notation, we will use the notation $h_x$ to not only mean a node
label, but also the node index/number ($x$).  Let the end
points of edges incident on \textbf{extremal} node on one
side, $a_{\mathbf{J}-1}$, be \textit{numbered} as \{ $h^{i}_{\mathbf{J}-1}$:
$0\leq\,i\,<\,\gamma$\}, where $h^{i}_{\mathbf{J}-1}$ are indices sorted in
\textit{increasing} order.  For each edge (fixed `\textit{i}') in
this set of edges of extremal node, $\langle a_{\mathbf{J}-1},
h^{i}_{\mathbf{J}-1}\rangle$,
there already \textit{exist} a \textit{shift-replicated} \textit{real edge} $\langle a_0,
(h^{i}_{\mathbf{J}-1} + 1)\mbox{-mod(\textbf{J})}\rangle$, and its further shift
replicas, in the \textbf{original} (unexpanded) graph. However, in general
for various numbers $h^{i}_{\mathbf{J}-1}$, \textbf{J} and
(non-zero) $\alpha$, and fixed `\textit{i}',
\[ (h^{i}_{\mathbf{J}-1} + \alpha +
1)\mbox{-mod(\textbf{J+$\alpha$})}\;\neq \; (h^{i}_{\mathbf{J}-1} + 1)\mbox{-mod(\textbf{J})} \]

\begin{algorithm}[!h]
\caption{Algorithm to `Expand' Order of a Circulant Balanced Bipartite Graph}
\label{prime_algo}
\begin{algorithmic}[1]
\algrenewcommand\algorithmicforall{\textbf{for each}}
   \State Label nodes of \textit{source} graph using sets \{$a_i$:$0\leq i<\mathbf{J}$\}
   and \{$h_i$:$0\leq i<\mathbf{J}$\}
    \State Label the edges of source graph using tuples $\langle a_{i},
    h^{k}_{i}\rangle$:
    $0 \leq i < \mathbf{J}$ and $0 \leq k < \gamma$

    \Statex
    \State Add $\alpha$ new nodes on either side towards making a bigger
    bipartite graph
    \State Label the newly added nodes with \{$a_i$:$\mathbf{J}\leq
    i<\mathbf{J+\alpha}$\} and \{$h_i$:$\mathbf{J}\leq
    i<\mathbf{J+\alpha}$\} respectively
    \State \textbf{Retain} all the edges, as represented by tuple of labels, in the
    bigger graph
    \Statex

    \ForAll{real edge in set $\langle a_{\mathbf{J}-1},
    h^{i}_{\mathbf{J}-1}\rangle$: $0 \leq
    i < \gamma$}
      \While{$1 \leq k < \mathbf{J}+\alpha$}
        \If{$\not\exists$ edge
                $\langle (a_{\mathbf{J}}+k-1)\mbox{-mod(\textbf{J}+$\alpha$)},
                (h^{i}_{\mathbf{J}-1}+k -1)\mbox{-mod(\textbf{J}+$\alpha$)}\rangle$}
          \State Add \textbf{dummy} edge $\langle
          (a_{\mathbf{J}}+k-1)\mbox{-mod(\textbf{J}+$\alpha$)},
        (h_{i,\mathbf{J}-1}+k)\mbox{-mod(\textbf{J}+$\alpha$)}\rangle$
        \EndIf
      \Statex
      \State k $\gets$ k+1
      \EndWhile
    \EndFor

    \Statex

    \ForAll{real edge in set $\langle a_{0}, h^{i}_{0}\rangle$: $0 \leq i < \gamma$}
      \While{$1 \leq k < \mathbf{J}+\alpha$}
        \If{$\not\exists$ edge $\langle a_{k},
        (h^{i}_{0}+k)\mbox{-mod(\textbf{J}+$\alpha$)}\rangle$}
          \State Add \textbf{dummy} edge $\langle a_{k},
          (h^{i}_{0}+k)\mbox{-mod(\textbf{J}+$\alpha$)}\rangle$
        \EndIf
      \Statex
      \State k $\gets$ k+1
      \EndWhile
    \EndFor
\end{algorithmic}
\end{algorithm}

In the above equation, the left hand side tries to coincide a
($\alpha+1)$-times circulantly shifted replica of edge $\langle
a_{\mathbf{J}-1}, h^{i}_{\mathbf{J}-1}\rangle$ in the expanded (bigger) graph, with the
existing edge $\langle a_0, (h^{i}_{\mathbf{J}-1} +
1)\mbox{-mod(\textbf{J})}\rangle$, the right hand side, which is not possible in general. Hence, in the expanded graph, where $\alpha$
dummy nodes have been added on either side of graph, the original, real
edge $\langle a_0, (h^{i}_{\mathbf{J}-1} + 1)\mbox{-mod(\textbf{J})}\rangle$
is \textit{no more a shift replica} of another real edge $\langle
a_{\mathbf{J}-1}, h^{i}_{\mathbf{J}-1}\rangle$. In fact, it \textit{may not
be shift replica} of \textit{any} original edge of
$a_{\mathbf{J}-1}$, $\langle a_{\mathbf{J}-1},
h^{k}_{\mathbf{J}-1}\rangle$.
\begin{equation}
\label{eqn_grow}
\forall k: 0 \leq i,k < \gamma\,:\; (h^{k}_{\mathbf{J}-1} + \alpha +
1)\mbox{-mod(\textbf{J+$\alpha$})}\;\neq \; (h^{i}_{\mathbf{J}-1} +
1)\mbox{-mod(\textbf{J})} 
\end{equation}
The shift-replication does hold in certain cases, in which case the
above equation becomes an equality. Let us define $|h_{\mathbf{J}-1} -
h^{i}_{\mathbf{J}-1}|$ as $d_i$. In the original graph, the real edge
$\langle a_0, (h^{i}_{\mathbf{J}-1} + 1)\mbox{-mod(\textbf{J})}\rangle$ is a
shift-replica of $i^{th}$ edge of $a_{\mathbf{J}-1}$, $\langle
a_{\mathbf{J}-1}, h^{i}_{\mathbf{J}-1}\rangle$. Then,
\textbf{whenever} $((h^{i}_{\mathbf{J}-1} + 1)\mbox{-mod(\textbf{J}}) -
d_i)\mbox{-mod(\textbf{J}}+\alpha)$ = $h^{k}_{\mathbf{J}-1}$ for some
\textbf{k} (may not be \textbf{i}), the former real edge continues to be
shift-replica of some earlier edge. For example, let $h^{i}_{\mathbf{J}-1}$
be equal to $h_{\mathbf{J}-1}$ (\textbf{k} = $\gamma$ - 1). It is
easy to see that $\langle a_0, h_0\rangle$ is still a ($\alpha+1$)-times
shift-replicated copy of $\langle a_{\mathbf{J}-1},
h_{\mathbf{J}-1}\rangle$, in the extended graph.
Otherwise, in general, the equivalence class of edges within a
perfect matching in context of earlier, smaller graph now breaks
down into \textbf{at maximum} two equivalence
classes. One equivalence class now contains the real edge (for fixed
`\textit{i}') $\langle a_{\mathbf{J}-1}, h^{i}_{\mathbf{J}-1}\rangle$, and their
shift-replicas in the bigger graph. The other equivalence class,
if needed, contains
another real edge (again, for fixed `\textit{i}') $\langle a_0, (h^{i}_{\mathbf{J}-1}
+ 1)\mbox{-mod(\textbf{J})}\rangle$, and their shift-replicas in the bigger
graph. Hence each node has upto $2\cdot\gamma$ (dummy+real) edges
incident on them, due to regularity of degree in the graph.

After partitioning each perfect matching, we \textbf{grow} each maximal
matching into a perfect matching of the extended graph by adding dummy
edges, which are shift replica of this class of edges. This leads to a
graph, which is circulant, but its node degree is \textbf{at maximum}
($2\cdot\gamma$). An example usage of such algorithm is depicted in figure
\ref{exp_prime_gph}, and summarized in algorithm \ref{prime_algo}. In this
figure, a order-5 bipartite graph (figure \ref{prime_gph}) is grown into
order-6 bipartite circulant graph. One can see that in the bigger
graph, edge $\langle a_0, ~h_2\rangle$ is \textbf{not} a shift replica of
any earlier existing edges, $\langle a_4,~h_4\rangle$, $\langle
a_4,~h_3\rangle$, $\langle a_4,~h_1\rangle$, as per equation
\ref{eqn_grow}. Hence we grow these edges separately to get two different
extended perfect matchings. While executing line (9) of above algorithm,
we add the shift-replicated edges.
\begin{itemize}
\item Dummy edge $\langle a_5,~h_5\rangle$ as shift replica of real edge $\langle a_4,~h_4\rangle$.
\item Dummy edges $\langle a_5,~h_4\rangle$, $\langle a_0,~h_5\rangle$ as
        shift-replicas of real edge $\langle a_4,~h_3\rangle$.
\item Dummy edges $\langle a_5,~h_2\rangle$, $\langle a_0,~h_3\rangle$,
        $\langle a_1,~h_4\rangle$, $\langle a_2,~h_5\rangle$ as
        shift-replicas of real edge $\langle a_4,~h_1\rangle$.
\end{itemize}
Similarly, while executing line (17) of the algorithm, we add the following
shift-replicated edges.
\begin{itemize}
\item Dummy edges $\langle a_3,~h_5\rangle$, $\langle a_4,~h_0\rangle$,
        $\langle a_5,~h_1\rangle$ as shift-replicas of real edge $\langle
        a_0,~h_2\rangle$.
\item Dummy edges $\langle a_1,~h_5\rangle$, $\langle a_2,~h_0\rangle$,
        $\langle a_3,~h_1\rangle$, $\langle a_4,~h_2\rangle$, $\langle
        a_5,~h_3\rangle$ as shift-replicas of
        real edge $\langle a_0,~h_4\rangle$.
\end{itemize}

\begin{figure}[h]
\begin{center}
\subfloat[Original Circulant
Graph]{\label{prime_gph}\includegraphics[scale=0.48]{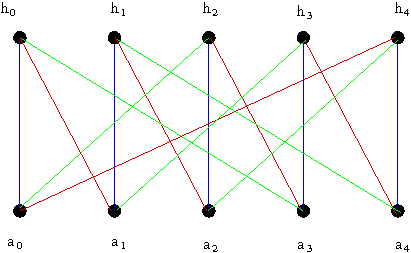}}
\qquad
\subfloat[Expanded Circulant
Graph]{\label{exp_prime_gph}\includegraphics[scale=0.48]{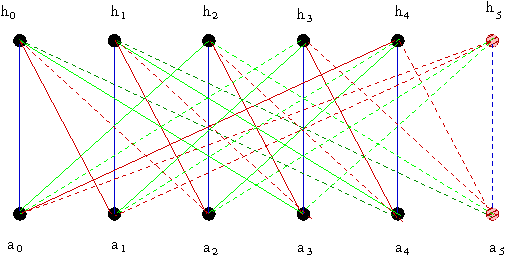}}
\end{center}
\caption{An Example Circulant Graph Expansion}
\label{graph_expand}
\end{figure}

A matrix version of above algorithm is described in \ref{mat_exp_sec}. It is easy to see that the overall graph is circulant with node degree 5,
as expected ($5 < (2\cdot\gamma\,=\,2\cdot 3\,=\,6)$. Also easy to see is that this
algorithm results in a bigger circulant bipartite balanced graph, which has
$\alpha$ additional dummy nodes on either side, and an at maximum
$\gamma$ additional dummy edges per real node.
All the edges added to the additional nodes are considered dummy edges,
since we do not intend to schedule any real computation on the additional
(dummy) node.

We now partition such a circulant graph and schedule the folding
in the standard way, as described in this paper. Whenever some dummy edges
incident on \textit{any} node are scheduled for input/output, they result
in dummy (no read/write) event. Theorem \ref{th1} holds, and the connection
remains static across folds, thus saving all the interconnect
reconfiguration time. This trades off with increase in the span of
the schedule, which is governed by the number of perfect access patterns
within the perfect sequence. In worst case, the number of perfect access
patterns, governed by ({\large $\lceil \frac{\gamma}{2}\rceil$}), grows by a factor
\textit{upto} 2. However, since we expect only small number of dummy nodes to
be added, the porosity of such schedule (no transmission/reception of data
on some edges in a particular machine cycle) will be less.
One can immediately see that \textbf{only} when
last fold is scheduled for computation, some of the PPUs are
idle during entire computation cycle of this fold.
Also, in the same fold, few
PMUs do not have any i/o scheduled at some of its ports,
in particular cycles. Hence some of the \textbf{full} (unfolded)
perfect access patterns are unbalanced in the last fold. For higher folding
factors \textbf{q}, such small imbalance is an acceptable part of our
design methodology.

\subsubsection{Expanding a Circulant Matrix}
\label{mat_exp_sec}
A circulant bipartite graph can also be represented in matrix from, via the
adjacency relation. The node indices of either side of bipartite graph form
the row and column indices of the matrix, respectively. If an edge exists
between two nodes, a 1 is present in corresponding place in the
matrix (0 otherwise). A $7\times 7$ circulant matrix representation of
bipartite graph of figure \ref{pg_bbg}, is shown in figure
\ref{fano_mat}.

\begin{figure}[h]
\begin{center}
\subfloat[Original Circulant
Matrix]{\label{fano_mat}\includegraphics[scale=0.88]{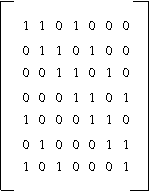}}
\qquad
\subfloat[Expanded Non-circulant
Matrix]{\label{fano_nonc_mat}\includegraphics[scale=0.79]{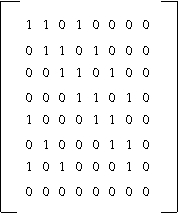}}
\qquad
\subfloat[Expanded Circulant
Matrix]{\label{fano_c_mat}\includegraphics[scale=0.79]{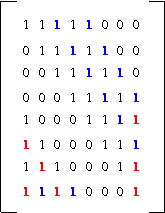}}
\end{center}
\caption{Adjacency Matrix of $7\times 7$ Geometry}
\label{appB_1_fig}
\end{figure}
One can see that in this matrix, if there is a `1' in position $\langle
i,j\rangle$, then there is a `1' again in position $\langle (i+1) \text{mod
7}, (j+1) \text{mod 7}\rangle$ (\uline{circulance property}). If we add a
row and a column having all `0's (equivalent of expanding the graph by
$\alpha$ = 1), the above property is no more valid; see figure
\ref{fano_nonc_mat}. Hence we need to \textit{overwrite} some `0's with
`1's in certain places, so that the above property holds again.

From figure figure \ref{fano_nonc_mat}, we see two sets of locations where
the circulance property is violated. For each `1' in last column of
original matrix ($a_{i,6}$ = 1), we find that certain $a_{(i+k)\text{-(mod
7)},(6+k)\text{-(mod 7)}}$ for $0 < k < 7-i$ are all `0'. We change such
`0's to `1's, as shown in red font in figure \ref{fano_c_mat}. Similarly,
for each `1' in first column of original matrix($a_{i,0}$ = 1), we find
that certain $a_{(i-k)\text{-(mod 7)},(7-k)\text{-(mod 7)}}$ for $0 < k \leq
i+1$ are all `0'. We change such `0's to `1's, as shown in blue font in
figure \ref{fano_c_mat}. This way, we complete all the principal and
non-principal diagonals having all values of `1'. It is easy to show that
this algorithm corresponds step-by-step to algorithm \ref{prime_algo}.

\subsection{Detailing Communication Architecture}
\label{comm_detail_sec}
At the next, \uline{\textbf{third} level of refinement}, we refine the
communication subsystem in the high-level architecture evolved in the
previous refinement. For this purpose, we \textit{expand} each edge in
Figure \ref{fold_arch}, and introduce two sets of 2-to-$\hat\rho$, and
$\hat\rho$-to-2 switches, and appropriate wiring between them. The
value of $\hat\rho$ is typically $\rho$ (see corollary \ref{corl1} for
definition of $\rho$).
Design details of these switches is discussed in
section \ref{switch_sec}. The wiring is \textit{governed}
by the generation of folded perfect access sequence generation, discussed
in section \ref{pat_sec}. The exact implementation of wiring can be guided by
details in section \ref{wire_sec}. At this level, the structural model of
the intended system is complete, and models for many intervals in its
overall \textit{cycle-accurate} behavior are also available.
This makes the system model at this level
\textit{approximately-timed}, as defined in \cite{gajski_tlm_pap}. The next
\textit{(fourth)
level of refinement} details and integrates such intervals, and completes
the entire \textit{cycle-accurate} schedule, and emitting the RTL model
thereafter.

The top level of complete
structure of the system is shown in figure \ref{full_arch_fig}. To
\textit{avoid congestion} in the diagram, the figure shows \textbf{only one
of the two} instances of the global, PG-based interconnect between
one of the two paired,
complementary sets \footnote{The set of 2-to-$\hat\rho$
switches on one side, and the set of
$\hat\rho$-to-2 switches on other side form a pair}
of these switches. This diagram is evolved for the
example system having 30 nodes, which was introduced as a running example
for entire section \ref{methodology_sec}, and for the fold factor discussed
in subsection \ref{abs_sched_sec}. The set of (5) edges having the
\textit{same color} reflect the fact that they are used in communication in
a \textit{synchronous way}. That is, in certain cycles, each of all the
edges/wires of a particular color (e.g., yellow), between two specific
ports of a pair of complementary switches carry data signals. The specific
connection details (which ports, which switches) are discussed in section
\ref{pat_sec}.

\begin{figure}[h]
\begin{center}
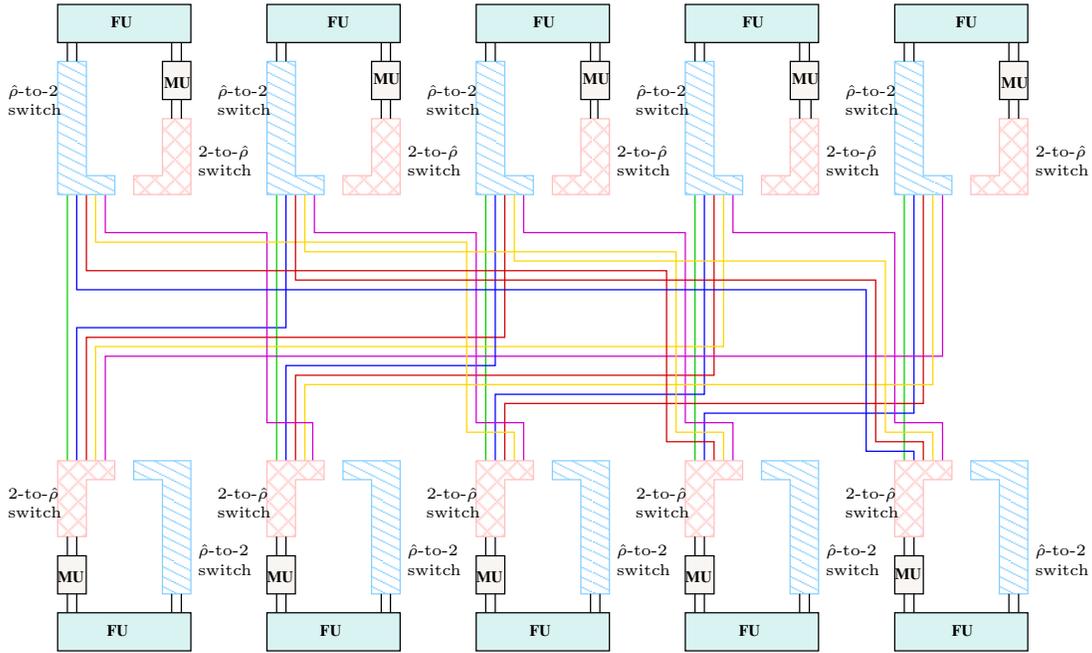
\end{center}
\caption{Top-level Completed Structure of Folded Systems with PG-based
Architectures}
\label{full_arch_fig}
\end{figure}

\subsubsection{The Structure of Switches}
\label{switch_sec}
2-to-$\hat{\rho}$ switches are used to interface the two
transmitting/output ports of each PMU, and the $\gamma$
\textit{possible} recipient/input ports of $\rho$ PPUs; see
corollary \ref{corl1}. Similarly, $\hat{\rho}$-to-2 switches are used to
interface the two receiving/input ports of each PPU, and the
$\gamma$ \textit{possible} transmitting/output ports of $\rho$ PMUs.
There are \textbf{two sets} of such 2-to-$\hat{\rho}$ and $\hat{\rho}$-to-2
switches, since there are two sets of PPUs/PMUs in the
high-level architecture. Regrouping these sets, there are
\textbf{two} paired, complementary sets of switches, where each paired set
consists of one out of two sets of 2-to-$\hat\rho$ switches
belonging to one side, and one out of two sets of $\hat\rho$-to-2
switches belonging to other side of the bipartite graph. Each such paired,
complementary set of switches is interconnected using
an instance of folded PG-based interconnect, as per section \ref{wire_sec}.
The selection bits for all of each type of switch, in each of the two sets,
in every relevant cycle, are synchronized and governed by calculations in
sections \ref{edge_map_sec} and \ref{mem_indx_sec}.

Mostly $\hat\rho$ is equal to
$\rho$ ($\hat\rho$ = $\rho$), but sometimes $\hat\rho$ $>$ $\rho$.
For details, the reader can skip to section \ref{wire_sec}.
In brief, for \textit{each perfect access pattern} whose
folding results in two node indices getting
re-mapped to same overlaid index, $a^{'}_{ij0}$ = $a^{'}_{ij1}$ as per section
\ref{wire_sec}, \textbf{one} \textit{additional} input/output port gets added to each switch
within the paired, complementary set of
switches to which the perfect pattern belongs. This
tantamounts to $\hat\rho$ = $\rho$ + $\theta$, where $\theta$ is the number
of perfect patterns for which $a^{'}_{ij0}$ = $a^{'}_{ij1}$. Each
perfect access pattern implies concurrent communication of two signals.
The additional port \textit{per such pattern} is needed in the
above case because two, rather than one, wires are needed to
\textbf{concurrently support} communication of two input signals between
every pair of matched 2-to-$\hat\rho$ and
$\hat\rho$-to-2 switch corresponding to the folded perfect access pattern;
again see section \ref{wire_sec}.

As pointed out in section \ref{sys_arch}, one type of PPUs are mapped to hyperplanes, and other type to points of a PG bipartite
graph. Correspondingly, when data is being read from PMUs
collocated with one type of PPUs, by the other type of
PPUs, then the 2-to-$\hat{\rho}$ switch, locally placed with PMUs, automatically assume the role of the PMU itself (point or
hyperplane). Similarly, $\hat{\rho}$-to-2 switch, locally placed with
PPUs, automatically assume the role of the PPU
itself (hyperplane or point).

Each switch can be implemented by putting its
port selection schedule in a
LUT, and driving a multiplexer/ demultiplexer from this LUT in appropriate
cycles. The schedule of \textbf{one} switch can be put in \textbf{one LUT},
and schedule of \textit{all other} switches of same type in the same set
can be derived using circulance property discussed in section
\ref{edge_map_sec}. The detailed scheduling of switches is discussed as
part of next level of refinement, in section \ref{comp_sched_sec}.

\subsubsection{Folded Perfect Access Sequence Generation}
\label{pat_sec}
The generation of \textit{folded perfect access
sequence} is one of the
\uline{most important step} towards defining the overall schedule for
system execution. This step leads to \textit{creation}, rather
than refinement, of a model of \textit{control flow} at the \textit{third
level of refinement}, since the required controls of datapath elements are
absent from system model so far. Thus, this model provides an abstract
view of communication scheduling. Generation of schedule is governed by
the details of the proof of theorem \ref{th1}, which in turn deals
with folding of a perfect access sequence. The model also provides
inputs about wiring: which 2-to-$\hat\rho$ switch to be wired to which
$\hat\rho$-to-2 switch, and between which two ports of such two switches.
These details will be brought out in later sections. From our design
experience, \textit{this abstract schedule is the most important input to the overall design
process}.

By using folded perfect access sequences, we can perform parallel
computation of individual nodes (PPUs) on one side of graph, in a multi-cycle
synchronous fashion as follows. As per our assumption about nature of
computation in section \ref{comp_model_sec}, we assume that the node computations
use only one occurrence of each input signal.

Whenever $\mathbf{p}$ is odd, then number of input/output
per computation, {\large $\mathbf{\gamma\,=\,\frac{p^{s n} - 1}{p^{s} - 1}}$},
is divisible by 2. Else, when \textbf{p} = 2, we add a \textbf{dummy edge} to each
node of one side in a \textit{circulant way}, with the edge ending in
\textit{any} node on the
other side. When physically scheduled, the communication over this edge, a
dummy read/write, results in \uline{no transaction}. Hence adding any
scheduling of such edges at various points of time in a \textit{balanced}
schedule leads to a \textit{balanced} schedule only. Physically,
we propose that individual nodes are designed to ignore such dummy input
value available at one of their ports, in the appropriate cycle, to avoid
miscomputation. After such addition, the new number of
input/output per computation is now divisible by 2. By taking, for example,
two inputs at a time for computation, we can periodically schedule a binary
operation on each PPU, in every few cycles (a sequential
computation may take more than one cycle). The set of two
edges representing the i/o for each node's current computation are chosen
so that the edge-pairs are \textit{shift
replicas} of one-another; see figure \ref{perf_pat}. In
\cite{karm1}, Karmarkar showed
that such 2-at-a-time processing indeed leads to perfect access pattern
generation. By folding the number of nodes,
and scheduling as per theorem \ref{th1}, we get \textbf{folded}
perfect access patterns for the folded architecture as well. Any
sequence of such folded perfect access patterns qualifies to be a folded
perfect access sequence. The algorithm for generation of
folded sequence is summarized in algorithm \ref{alg1}.

There is thus a three-level symmetry in computation scheduling
that we evolve. While exciting 2 inputs at a time, each group of \textbf{J/q}
PPUs belonging to one fold shows memory access
balance within a single cycle. Across
\textbf{q} such cycles, all the \textbf{q} groups show balance. These
balanced patterns from these \textbf{q} cycles combine to form a perfect
pattern, when combined \textit{temporally}. Finally, all such (combined) perfect
patterns should form a balanced perfect sequence. The execution of
perfect sequence, thus, takes multiple cycles.

\begin{algorithm}[h]
        \caption{Folded Perfect Access Pattern Generation}
\label{alg1}
\begin{algorithmic}[0]
\If{ $\gamma$ is odd}
    \State add an arbitrary dummy edge to each node, in a
    \textit{circulant} fashion
\EndIf
\Statex
\While{$\exists$ 2 more edges per node on one side of unfolded graph}
    \ForAll{ $0 \leq i < q$ folds of graph}
      \ForAll{ node $h_{ij}$: $j^{th}$ node on one side in $i^{th}$ fold of
graph}
        \State Select 2 so-far \textit{unselected} edges of $h_{ij}$, related to previous
    considered node in a \textit{circulant} fashion, \\
    ~~~~~~~~~~~~~~~~~~~~$e_{ijk}$ and $e_{ijl}$ \\
            \Comment The selection depends on order of inputs as required by node computations
        \State Calculate their new end points as follows\\
    ~~~~~~~~~~~~~~~~ $a^{'}_{ijk}$ = $a_{ijk}$ mod-(\textbf{J/q}) \\
    ~~~~~~~~~~~~~~~~ $a^{'}_{ijl}$ =
    $a_{ijl}$ mod-(\textbf{J/q})
      \EndFor
      \Statex
      \State Perfect Access Pattern = \{ $\langle$
      $h_{ij}$ mod \textbf{J/q}, \{$a^{'}_{ijk}$, $a^{'}_{ijl}$
\} $\rangle$, $\ldots$ \} $\forall$ $0 \leq j < \mathbf{J/q}$, $0 \leq
k,l < \mathbf{\gamma}$
    \EndFor
    \Statex
    \State \textbf{Full} Perfect Access Pattern = Sequence of above perfect
    patterns $\forall$ $0 \leq i < \mathbf{q}$
\EndWhile
\Statex
\State Perfect Sequence = Sequence of above \textbf{Full} Perfect Access Patterns
\end{algorithmic}
\end{algorithm}

An \textit{important} \textbf{2-way design option} for folded architectures
is as follows. There are two ways by which we can combine the 2-input 
computations done by nodes of a fold. We may \textit{first} schedule
2-input computations to be done by each of the
\textbf{J} nodes across \textbf{all} the \textbf{q}
folds sequentially, and \textit{then} we combine partial
all such partial schedules into \textbf{full}/unfolded perfect access patterns.
Alternatively, we may \textit{first} sequentially schedule all
$\gamma/2$ 2-input computations done by each of
the \textbf{J/q} nodes in one fold only, and \textit{then}
repeat this schedule
for all remaining \textbf{(q-1)} folds, and finally combine such
patterns. The choice of
this is left to the implementer. For deciding schedules of various
components, we will \textit{use first design option hereafter}, unless stated
otherwise.

\subsubsection{Example Folding and Abstract Schedule Generation}
\label{abs_sched_sec}
Any sequence of perfect access patterns computed in section \ref{pat_sec}
gives rise to an \textit{abstract} version of computation and communication
schedule. We describe this abstract schedule by folding the example graph
of table \ref{h_p_ex_tab}.

For that graph, we can fold the 15 nodes on each side by a factor of 3, so
that each fold/partition has 5 nodes of either type. Running the algorithm
\ref{alg1}, we get the schedule as in table \ref{fold_ex_tab}. The 15
LPUs are been referred as PUs, 5 physically used PPUs as
PUs, and 5 physically used
PMUs as MUs. A dummy MU is used as a placeholder in last perfect access pattern
for the no memory transaction that is to be scheduled on 2nd port of a
PU.

\begin{table}[!h]
\caption{An Example Folding Schedule. \textbf{D} implies Dummy Edge}
\label{fold_ex_tab}
\centering
{\scriptsize
\begin{tabular}[!h]{|p{0.4cm}|p{2cm}|p{2cm}|p{2cm}|p{2cm}|p{2cm}|p{2.3cm}|}
\hline
Cycle \# & \multicolumn{6}{|c|}{Folded Pattern} \\ \hline \hline 
\multicolumn{7}{|c|}{\textbf{Full} Perfect Access Pattern 0} \\ \hline \hline 
 
0 & [PU0 : MU0, MU1 ] & [PU1 : MU1, MU2 ] & [PU2 : MU2, MU3 ] & [PU3 : MU3, MU4 ] & [PU4 : MU4, MU0 ]  & Scheduling $0^{th}$, $1^{st}$ edge of {0,1,2,3,4} PUs \\ \hline \hline 
1 & [PU0 : MU0, MU1 ] & [PU1 : MU1, MU2 ] & [PU2 : MU2, MU3 ] & [PU3 : MU3, MU4 ] & [PU4 : MU4, MU0 ]  & Scheduling $0^{th}$, $1^{st}$ edge of {5,6,7,8,9} PUs \\ \hline 
2 & [PU0 : MU0, MU1 ] & [PU1 : MU1, MU2 ] & [PU2 : MU2, MU3 ] & [PU3 : MU3, MU4 ] & [PU4 : MU4, MU0 ]  & Scheduling $0^{th}$, $1^{st}$ edge of {10,11,12,13,14} PUs \\ \hline \hline 
\multicolumn{7}{|c|}{\textbf{Full} Perfect Access Pattern 1} \\ \hline \hline 
 
3 & [PU0 : MU2, MU4 ] & [PU1 : MU3, MU0 ] & [PU2 : MU4, MU1 ] & [PU3 : MU0, MU2 ] & [PU4 : MU1, MU3 ]  & Scheduling $2^{nd}$, $3^{rd}$ edge of {0,1,2,3,4} PUs \\ \hline 
4 & [PU0 : MU2, MU4 ] & [PU1 : MU3, MU0 ] & [PU2 : MU4, MU1 ] & [PU3 : MU0, MU2 ] & [PU4 : MU1, MU3 ]  & Scheduling $2^{nd}$, $3^{rd}$ edge of {5,6,7,8,9} PUs \\ \hline 
5 & [PU0 : MU2, MU4 ] & [PU1 : MU3, MU0 ] & [PU2 : MU4, MU1 ] & [PU3 : MU0, MU2 ] & [PU4 : MU1, MU3 ]  & Scheduling $2^{nd}$, $3^{rd}$ edge of {10,11,12,13,14} PUs \\ \hline \hline 
\multicolumn{7}{|c|}{\textbf{Full} Perfect Access Pattern 2} \\ \hline \hline 
 
6 & [PU0 : MU0, MU3 ] & [PU1 : MU1, MU4 ] & [PU2 : MU2, MU0 ] & [PU3 : MU3, MU1 ] & [PU4 : MU4, MU2 ]  & Scheduling $4^{th}$, $5^{th}$ edge of {0,1,2,3,4} PUs \\ \hline 
7 & [PU0 : MU0, MU3 ] & [PU1 : MU1, MU4 ] & [PU2 : MU2, MU0 ] & [PU3 : MU3, MU1 ] & [PU4 : MU4, MU2 ]  & Scheduling $4^{th}$, $5^{th}$ edge of {5,6,7,8,9} PUs \\ \hline 
8 & [PU0 : MU0, MU3 ] & [PU1 : MU1, MU4 ] & [PU2 : MU2, MU0 ] & [PU3 : MU3, MU1 ] & [PU4 : MU4, MU2 ]  & Scheduling $4^{th}$, $5^{th}$ edge of {10,11,12,13,14} PUs \\ \hline \hline 
 
\multicolumn{7}{|c|}{\textbf{Full} Perfect Access Pattern 3} \\ \hline \hline 
9 & [PU0 : MU0, D ] & [PU1 : MU1, D ] & [PU2 : MU2, D ] & [PU3 : MU3, D ] & [PU4 : MU4, D ] &  Scheduling $6^{th}$ edge of {0,1,2,3,4} PUs \\ \hline 
10 & [PU0 : MU0, D ] & [PU1 : MU1, D ] & [PU2 : MU2, D ] & [PU3 : MU3, D ] & [PU4 : MU4, D ] & Scheduling $6^{th}$ edge of {5,6,7,8,9} PUs \\ \hline 
11 & [PU0 : MU0, D ] & [PU1 : MU1, D ] & [PU2 : MU2, D ] & [PU3 : MU3, D ] & [PU4 : MU4, D ] & Scheduling $6^{th}$ edge of {10,11,12,13,14} PUs \\ \hline

\end{tabular}
}
\end{table}

The schedule of PUs in each fold per clock cycle can be easily seen to be
balanced. Put together, they first form a \textbf{full} perfect access pattern every
3 cycles, and then perfect access sequence in 12 cycles.

\subsubsection{Wiring the Interconnect}
\label{wire_sec}
As mentioned earlier, wiring is assumed to be \textbf{direct} in our
case. By theorem \ref{th2}, it is possible to fold in such a way that
certain (overlaid) nodes always access \textbf{same set} of $\rho$ out of
\textbf{J/q} PMUs. Hence the \textbf{connections remain static},
as the computation schedule moves from one fold to another. This is one of
the \textbf{most significant advantages} of folded PG bipartite graphs.
Each wire connects one port of a 2-to-$\hat{\rho}$
switch, and one port of a $\hat{\rho}$-to-2 switch, as already discussed in section \ref{switch_sec}. This
static-ness is easily illustrated using the example folding shown in table
\ref{fold_ex_tab}, by picking any column and each set of 3 continuous
rows under some \textbf{full} perfect access pattern.

Referring to section \ref{prob_form_sec}, if the end points of two
connections of a particular node being considered in a particular
cycle, in a folded graph are
equal (e.g. $a^{'}_{000}$ = $a^{'}_{001}$), the number of wires to each
PMU from each reachable PPU become double. It requires
double channel width, which trades off with decrease in the switch size.
Also, wiring two interconnects between same pair of source and
destination nodes may possibly lead to subsequent
wiring/routing congestion at later design flow stages. One can then
alternatively try to design for another folding factor. Since our
methodology accepts any \textbf{q} that is a factor of \textbf{J}, we
can vary \textbf{q} and may get
a design for which $a^{'}_{000}$ $\neq$ $a^{'}_{001}$.

\subsubsection{Relating Communication Refinement to Modification in
Microarchitecture of PPUs}
\label{micro_mod_sec}
The \textit{fundamental} problem of \textbf{overlaying of datapath
elements} needs to be handled in all possible folding designs.
This design step naturally fits in the \textit{second level of refinement}, which
deals with computational refinement. Hence it has been handled via
creation of the untimed model. However, timing of this model depends on order of input arrival,
i.e. the choice of a design option discussed in section \ref{pat_sec}. Hence this part of micro-architecture evolution is
made part of \textit{third level of refinement}.

Especially in case of operators, within PPUs, that \textit{consult all
input data} to a node('s computation), some changes are needed to
save state, including the intermediate results. For example, let each node's
computation have an accumulation(/max/min) operator present
within. In the schedule of first folding design option, accumulation is only
done partially for each node that is overlaid on the PPU,
across multiple folds during one run of a perfect access pattern per fold.
The current
partial sum needs to be stored separately for each fold, since in the next
run of perfect access pattern in the sequence for
the same fold multiple cycles later, this partial sum needs to
be carried over. Hence per PPU, $q$ copies of each register
holding such intermediate result need to be created.

In the second design option, any register along the
datapath of PPU, whose contents are read and used later
on after multiple cycles, needs to  again have $q$
copies each. This is because in this interval, overlay of such register
would have happened. Of course, switches to select the right register
copy in a particular cycle, driven by the fold index
currently in operation, also need to be inserted in the datapath for this
design option.

\subsection{Issues in Overall Scheduling and Design Completion}
\label{issues_sec}
The control path of a synchronous VLSI system is implemented using a
cycle-level schedule. All aspects of
folding being dealt in the current section \ref{methodology_sec} pertain to
folding the data path of a suitable system, by doing stepwise
refinement of the corresponding DFG. The control path can be 
evolved \textit{alongside}, from the original schedule of an unfolded VLSI system. In the
schedule of such system, there will be \textbf{intervals}, in which
datapath elements will be re-used. By
interval, we imply some contiguous sequence of machine cycles. Such intervals need to be
\textit{expanded by a factor}, along with insertion of \textit{new control
signals} which define e.g. the fold index currently in operation.
Expanding generally implies replicating an interval in which a
certain control signal is TRUE, \textbf{q} times in a contiguous way.
Memory access interval, node computation interval, switch
enable intervals etc. all need to be expanded by a factor. It is possible to
identify and enlist such intervals at RTL level model of the datapath.
Automating the generation of new, expanded schedule using this list,
especially when control path is implemented using microcode sequencing, is
straightforward.

However, some of these expansions can be best worked out from
scratch, rather than working with an interval of schedule for the unfolded
system. This is because in some places, rather than interval of one signal,
interval of a set of related signals gets expanded by factor
\textbf{q}. Further, in such groups, the order in which signals were
earlier turned TRUE gets \textbf{rearranged}. For example, group of switch
selection signals show this characteristic due to folding. Hence it was pointed out earlier that
after the \textit{third level of refinement}, intervals in the
\textit{cycle-accurate} behavior of the intended system, some reflecting
folding and others not reflecting folding, are also available. For such
intervals, the schedule generator must focus on inserting/replacing
appropriate schedule intervals, rather than expanding. To generate
such replacement intervals, the schedule derived in section
\ref{abs_sched_sec} is used as \textbf{base} schedule to derive individual
schedules (cycle-accurate behaviors). To summarize, it is the \uline{\textbf{fourth} level of refinement}
that expands/inserts and integrates these intervals, completing
the implementation of entire
control path of the system via a \textit{cycle-accurate}
schedule (system behavior), and emitting the RTL model thereafter.

Though this schedule governs
the behavior of individual components, certain \textbf{auxiliary details} such as
selection order of ports of
some switches, which is needed for schedule derivation,
also need to be now specified. We cover all these detailed
auxiliary issues, and the overall schedule derivation, in
remaining part of section \ref{methodology_sec}. Before going into
details, we first summarize all the remaining issues that need to be
tackled. Generating details corresponding to solution of these issues
is the other concern of the \textit{fourth level of refinement}.

A schedule for the parallel computational model discussed in section
\ref{comp_model_sec} needs to address issues in two identical computation
phases, due to flooding nature of the computation algorithm. 
Correspondingly, as shown in section \ref{sys_arch}, there is a pair of
$\langle$\textit{PPU, PMU}$\rangle$ relations. One
relation relates
PPUs of left side of bipartite graph to PMUs on right
side of bipartite graph, from which they read the input data in parallel.
Similarly, the other relation relates PPUs of right side of
bipartite graph to PMUs on left side of bipartite graph, from
which they read the input data in parallel. The two reading phases, though
identical, are \textbf{disjoint}. Hence we can simply solve the issues in
communication schedule derivation for one relation only, and apply the
answers to the other.

We identify the following issues in generating the communication schedule.

\subsubsection{Issues for Physical Processing Units}
\label{fissues_sec}
For a \textbf{full} (non-folded) perfect access pattern, after folding, we note the
following issues.
\begin{enumerate}
\item Each LPU, when scheduled over an
overlaid PPU, reads two data items from two of its edges in a
particular machine cycle. How to know which two edges are being active?
\item The $i^{th}$ one out of (\textbf{J/q}) PPUs of $k^{th}$ fold
accesses \textit{one or both} its data in $p^{th}$ PMU for the $l^{th}$
perfect access pattern (see theorem \ref{th1}). How to get the value of \textbf{p}?
\item How to decide whether one or both the data are going to be
stored/read in the same PMU?
\item Given the index of PMU, from which
locations will one/both of the data items be read during $l^{th}$
\textbf{full} perfect access pattern?
\setcounter{saveissueFi}{\theenumi}
\end{enumerate}
The last issue actually pertains to address generation for the read data.
Hence we address this issue as part of the issues in PMU
scheduling itself, in the next section.

Since after computation, PPUs write the result in their local
memory, there are no folding-related issues in
write-back. This is data is to later read by PPUs
of the opposite side, using the edge/connection that connects the PPU and
the PMU. Two issues for a PPU, while writing back data corresponding to an
edge, are:
\begin{enumerate}
\setcounter{enumi}{\thesaveissueFi}
\item After computation, at which location of local memory must each
PPU write the data corresponding to an edge?
\item At each location, in which machine cycle must each PPU
      write the corresponding data?
\end{enumerate}

\subsubsection{Issues for Physical Memory Units}
\label{missues_sec}
The PMUs are also
involved in distributing read data in parallel to various
PPUs. The reading of data is in bursts, and it happens in
certain successive cycles that make up the entire perfect access sequence.
Correspondingly, read addresses need to be generated somewhere in the
system, which are used by PMUs to provide data in various machine
cycles.

For a \textbf{full} (non-folded) perfect access pattern, after folding, we note the
following issues.
\begin{enumerate}
        \item To which PPUs must a PMU send out data?
\setcounter{saveissueMi}{\theenumi}
\\
This question is a dual question of $2^{nd}$
issue for PPUs, and can be easily solved for by
inverting the map generated for that problem. Hence we leave out
reporting detailed solution to this issue.
\end{enumerate}
\begin{enumerate}
\setcounter{enumi}{\thesaveissueMi}
\item In a given cycle, a PMU must send out
data from which location, to which PPU? \\
Because this issue is dealt by generating
corresponding address, we \textbf{transform} this question into following
\textit{address generation issue}. If the PPU $h^{m0}$ working
on some binary operation (read-)accesses the $\mathbf{m}^{th}$ PMU, then in which
cycle does it access it, and at which location (local address)? Here,
$h^{m0}$ is defined as the node of the unfolded graph, whose location
on one side of the bipartite graph is
\textbf{extremal} w.r.t. other
connected nodes to $\mathbf{m}^{th}$ PMU. Answering this question,
and then extending the schedule using the sequence generation implicit in
section \ref{pat_sec}, the entire addressing can in fact be evolved.
\end{enumerate}

Another set of issues arise, when addresses need to be generated
for local memory during the write-back phase of a PPU. In this
phase, the PMU is fixed: it is the local memory. However, the
location in which a datum must be written in each cycle varies. It is easy
to notice that this issue is addressed by the last two (address generation)
issues in section \ref{fissues_sec}. The order in which PPUs of other side/type will access datum for input dictates the order in
which data must be stored into these local memories. The read/write address generation
issues will hence be address \textit{jointly} later.

Throughout remaining section, we \textit{continue} to assume the
natural left-to-right labeling of vertices on either side of the
graph, as shown in figure \ref{perf_pat}.

\subsection{Solutions to Auxiliary Issues}
\label{aux_issue_sec}
The detailed solutions to above issues are discussed in this
section. A reader may choose to skip over to next section
\ref{comp_sched_sec} during initial reading.

\subsubsection{Edges used in a Perfect Access Pattern}
\label{edge_map_sec}
In this section, we address the $1^{st}$ issue raised in section
\ref{fissues_sec}. To summarize, this issue relates to finding out which
two edges of each node of the folded graph will be used for reading data
in a particular cycle. Recall from section \ref{comm_detail_sec} that
2-to-$\hat{\rho}$ switches are interfaced with output ports of various MUs.
Addressing $1^{st}$ issue is important to \textbf{synchronize} the
port selection logic of all
2-to-$\hat{\rho}$ switches, that are interfaced to PMUs of each
type.
This is because the switches address their lines in a \textbf{local} way,
i.e. labeling of their output ports is local. One has to then provide an explicit
mapping so that the \textit{local indices} of lines selected by e.g.
2-to-$\hat{\rho}$ demultiplexer switches, present at the output of each PMU, form an (unfolded) perfect access pattern. It also
completes the behavioral specification of 2-to-$\hat{\rho}$ demultiplexer
switches.

PMUs are themselves responsible for generating the
port selection bits, to be used in various perfect access patterns.
Partitioning the edge set into subsets of two, and sequencing of these subsets,
for each set of two folded PG
interconnects, as defined in section \ref{pat_sec}, is needed to define
these patterns within a perfect access sequence for each of these sets. The
address generation has been covered in detail in section \ref{add_gen_sec}
later. The interconnect connects either hyperplane nodes to point
nodes, or point nodes to hyperplane nodes, depending which of the two
folded PG interconnects we are working with.
Correspondingly, the synchronized scheduling of ports of
2-to-$\hat{\rho}$ switch is
based on partitioning either the sorted point set of the hyperplane (index)
corresponding to the switch, or the hyperplane set of the point (index)
corresponding to the switch, whichever is the role of the switch (also see
table \ref{fold_ex_tab}). Either way, each
PPU receives two data input on two edges. Given a PMU (and a local 2-to-$\hat\rho$ switch) with
index \textbf{m}, we consider the left-extremal node
(corresponding to a $\hat\rho$-to-2 switch) connected to it in the
\textit{unfolded graph}, $h^{m0}$. Here, extremality implies that the
location of $h^{m0}$ on one side of the unfolded bipartite
graph is in \textit{left extreme} w.r.t.  other connected nodes to
$\mathbf{m}^{th}$ PMU. For example, in figure
\ref{pg_bbg}, node \textbf{p2} is \textbf{extremally} connected to
node \textbf{l1}. Further, let the \textit{totally ordered}
point set of $h^{m0}$ be denoted as \{$a^{m0}_{0}$, $a^{m0}_{1}$, $\ldots$, $a^{m0}_{(\gamma
-1)}$\}, where $a^{m0}_{0}$ $<$ $a^{m0}_{1}$ $<$ $\ldots$ $<$ $a^{m0}_{(\gamma
-1)}$. Let us also \textbf{impose} an order on the edges of
$h^{m0}$, so that we define the $r^{th}$ data of $h^{m0}$ to be the edge
between $h^{m0}$ and $a^{m0}_{r}$.

However, while the $r^{th}$ data of $h^{m0}$ is $r^{th}$ leftmost or
rightmost edge of $h^{m0}$, it may not correspond to the $r^{th}$ leftmost
or rightmost edge for $h^{mi}$, due to circulant rotation applied on the
edges. Here, finding $r^{th}$ leftmost or rightmost edge of a node
corresponds to sorting the destination nodes of various edges incident on
the source node, in increasing order, and taking the $r^{th}$ element of
sequence and its corresponding edge, \textbf{exactly} as discussed in
previous paragraph. Hence we need to
have a way, which given an edge, provides which all edges are circulant
shift-replicas of it. We give the details of such circulant edge
mapping now.

Recall that $h^{m0}$ $\equiv$ \{$a^{m0}_{0}$, $a^{m0}_{1}$, $\ldots$, 
$a^{m0}_{(\gamma -1)}$: $a^{m0}_{0}$ $<$ $a^{m0}_{1}$ $<$ $\ldots$ $<$ 
$a^{m0}_{(\gamma -1)}$\} (ordered point set).
Hence \textbf{m} is equal to $a^{m0}_{t}$ for some \textbf{t}. Let us
take another \textit{arbitrary} node $h_i$, which may or may not be
connected to the PMU \textbf{m}. Without loss of generality, let
($h_{i}$ -- $h^{m0}$) = $d_i$, where the difference is taken
modulo-\textbf{J}, and hence is always positive.
Then, due to circulance, the point set of $h_{i}$ can be
represented as \{$a^{m0}_{0}+d_i$, $a^{m0}_{1}+d_i$, $\ldots$, $a^{m0}_{(\gamma
-1)}+d_i$\}. The addition here is again modulo-(\textbf{J}) addition. Because
of modulo addition, the total order $a^{m0}_{0}$, $a^{m0}_{1}$, $\ldots$,
$a^{m0}_{(\gamma -1)}$ gets shifted in a circular way over the modulo `ring'.
If we sort this set of
indices in increasing order, then \{$\hat{a}^{m0}_{0}$ ( = $a^{m0}_{0}+d_i$),
$\hat{a}^{m0}_{1}$ ( = $a^{m0}_{1}+d_i$), $\ldots$,
$\hat{a}^{m0}_{(\gamma -1)}$ ( = $a^{m0}_{(\gamma -1)}+d_i$)\}, must
be \textbf{equivalent to} $\hat{a}^{m0}_{x}$ $<$ $\hat{a}^{m0}_{(x+1)}$ $<$ $\ldots$ $<$
$\hat{a}^{m0}_{(\gamma -1)}$ $<$ $\hat{a}^{m0}_{0}$ $<$ $\hat{a}^{m0}_{1}$ $<$ $\ldots$ $<$
$\hat{a}^{m0}_{(x-1)}$ for some \textbf{x}. It can easily be verified now that if
the edge between \textbf{m} and $h^{m0}$ was $r^{th}$ edge of $h^{m0}$,
then the corresponding shift-replicated edge incident on $h_i$ is an edge
between $h_i$ and $\hat{a}^{m0}_{(r-1)}$. This edge \textbf{need not} be the $r^{th}$
element of the sequence $\hat{a}^{m0}_{x}$ $<$ $\hat{a}^{m0}_{(x+1)}$ $<$ $\ldots$ $<$
$\hat{a}^{m0}_{(\gamma -1)}$ $<$ $\hat{a}^{m0}_{0}$ $<$ $\hat{a}^{m0}_{1}$ $<$ $\ldots$ $<$
$\hat{a}^{m0}_{(x-1)}$.

As an example, we take the graph of table \ref{h_p_ex_tab}. Let
\textbf{m} be $6^{th}$ point, i.e. \textbf{p6}. From the table, its
left-extremal neighboring hyperplane is \textbf{h1}. Let \textbf{r} = 4, in
which case the $4^{th}$ edge of \textbf{h1} connects \textbf{h1} and
\textbf{p5} (not \textbf{p6}). Let $h^{mi}$ = $\mathbf{h_{12}}$, in which
case $d_i$ = 11. In terms of total order, the $4^{th}$
left-to-right edge of $\mathbf{h_{12}}$ ends on \textbf{p7}, but this edge
is \textit{not} a shift-replica of the edge $\langle \mathbf{h1},
\mathbf{p5}\rangle$. Rather, the $4^{th}$ edge of
$\mathbf{h_{12}}$, which should be a
shift-replica of $4^{th}$ edge of $\mathbf{h_{1}}$, runs between
$\mathbf{h_{12}}$ and $\mathbf{p_{\left((5+11)\mod 15\right)}}$ =
$\mathbf{p_{1}}$.  Looking at the table, we find that this is indeed true.

An LUT can be used to store this
edge-selection schedule. A simple way of generating the edge-correspondence
is to start by choosing an \textbf{m} such that $h^{m0}$ has a label of 0.
Defining an order on edges on $0^{th}$ node is then natural,
straightforward left-to-right labeling.
 
For some designs, in the last perfect access pattern, a dummy edge
is scheduled, to allow a PPU to read from a dummy
PMU, a \textit{no value}.
To implement this, selection of dummy MU for input to a
2-to-$\hat\rho$ switch is done by using an invalid value of selection
signal, so that all but one output of 2-to-$\hat\rho$ switch remain
tristated, thus achieving the effect of \textit{no value read} on one
port.

\subsubsection{Pairing PPUs with PMUs}
\label{mem_indx_sec}
In this section, we address $2^{nd}$ and $3^{rd}$ issues raised in section
\ref{fissues_sec}. To summarize, the former issue relates to finding the
PMUs to be contacted while execution of a particular \textbf{full}
perfect access pattern, while the latter issue relates to knowing if both
the data are to be read from single PMU. Like in
previous section,
addressing these issues is important to \textbf{synchronize} the
port selection logic of all $\hat{\rho}$-to-2 switches that are collocated
with PPUs of each type, for \textbf{each} perfect pattern
within the sequence
evolved in section \ref{pat_sec}. Hence, hereafter we will address
the issue of synchronizing $\hat{\rho}$-to-2 switches for \textit{any}
perfect access pattern, by using a \textit{variable} index.
Like 2-to-$\hat\rho$
switches, these switches also address their lines in
a \textbf{local} way, i.e. labeling of their input ports
is local. One has to
then provide an explicit mapping so that the \textit{local indices} of
lines selected by e.g. $\hat{\rho}$-to-2 multiplexer switches, present at
the input of each PPU, form the \textbf{same}
\textit{folded}
perfect access pattern, that the PMUs of other type
use for communication, as per
previous section \ref{edge_map_sec}. Since the two chosen ports of all
2-to-$\hat\rho$ switches of are synchronized, it is \textbf{necessary to ensure}
that the set of \textbf{destination ports} of wires stimulated during execution of a
particular perfect access pattern, and having source ports in the
2-to-$\hat\rho$ switches, are specifically assigned based on the synchronized
choice of two ports made on all of the $\hat\rho$-to-2 switches.
\textbf{Only} that way, the signal driven on a wire by e.g.
2-to-$\hat{\rho}$ switch, will pass through a selected port of
$\hat{\rho}$-to-2 switch in next cycle, towards the destined PPU. Yet again, making such selections for all patterns in a communication
sequence also completes the behavioral specification of $\hat{\rho}$-to-2
multiplexer switches.

Overall, the synchronized scheduling of ports of $\hat{\rho}$-to-2
switches for entire perfect access sequence is done by
using schedule \textbf{reciprocal} to that of schedule of ports of
2-to-$\hat{\rho}$ switches. In an \textit{unfolded} design, it is easy to
prove that this can be obtained by doing same partitioning of
hyperplane/point set corresponding to each $\hat{\rho}$-to-2
switch, but by inverting the sorted order of the set first. However, it is
not straightforward in a folded design to get the inverse schedule in such
easy way. Hence we derive the inversion by first principles as follows.

To know the contacted PMUs by a PPU, as
enabled by the contact of various $\hat\rho$-to-2 switches with corresponding
2-to-$\hat\rho$ switch, we first try to
calculate the value of \textbf{p}, where
$i^{th}$ one out of (\textbf{J/q}) PPUs of $k^{th}$ fold
accesses one or both its data in $p^{th}$ PMU for $l^{th}$
perfect access pattern. Here, $l^{th}$ perfect access pattern is defined as one
that executes $(2\star l)^{th}$ and $(2\star l +1)^{th}$
edges of $0^{th}$ PMU; see section
\ref{edge_map_sec} (and table \ref{fold_ex_tab} for an example).

\begin{algorithm}[!h]
\caption{Memory Unit Assignment}
\label{alg2}
\begin{algorithmic}[0]
\ForAll{ fold index $k$, $0 \leq k < q$}
    \ForAll{ node $h_{ik}$ in the $k^{th}$ fold of folded graph, $0 \leq i < J/q$}
        \\\Comment Let $h^{m0}$ be an extremal node of unfolded graph
        connected to some memory $m$: preferably $h^{m0}$ = $0$
        \State $h^{m0}$ $\equiv$ \{$a^{m0}_{0}$, $a^{m0}_{1}$, $\ldots$, $a^{m0}_{(\gamma
-1)}$\}: $a^{m0}_{0}$ $<$ $a^{m0}_{1}$ $<$ $\ldots$ $<$ $a^{m0}_{(\gamma
-1)}$
        \State  $d_{ik}$ $\gets$ ($h_{ik}$ -- $h^{m0}$)
        \ForAll{$l^{th}$ perfect access pattern executing on node $h_{ik}$, $0 \leq l < \gamma/2$}
          \State $p$ $\gets$ [($a^{m0}_{(2l)}+d_{ik}$)
        modulo-(\textbf{J})] modulo-(\textbf{J/q})
          \State $\hat{p}$ $\gets$
          [($a^{m0}_{(2l+1)}+d_{ik}$)
        modulo-(\textbf{J})] modulo-(\textbf{J/q})
        \EndFor
    \EndFor
\EndFor
\end{algorithmic}
\end{algorithm}

We use the non-folded regular bipartite graph to answer this. We also use
the correlation between edges belonging to same perfect access pattern, brought
out in previous section \ref{edge_map_sec}. Given a PMU
index \textbf{m} and the extremal node connected to it in the
\textit{unfolded graph}, $h^{m0}$, let its \textit{totally ordered}
point set be denoted as \{$a^{m0}_{0}$, $a^{m0}_{1}$, $\ldots$, $a^{m0}_{(\gamma
-1)}$\}, where $a^{m0}_{0}$ $<$ $a^{m0}_{1}$ $<$ $\ldots$ $<$ $a^{m0}_{(\gamma
-1)}$. For the $h_{ik}$ = {\large $\left(k\cdot \frac{\mathbf{J}}{q} + i\right)^{th}$} node in the
unfolded graph, let ($h_{ik}$ -- $h^{m0}$) = $d_{ik}$, where the difference is taken
modulo-\textbf{J}. Due to circulance, the point set of $h_{ik}$ can be
represented as \{$a^{m0}_{0}+d_{ik}$, $a^{m0}_{1}+d_{ik}$, $\ldots$, $a^{m0}_{(\gamma
-1)}+d_{ik}$\}. The addition here is again modulo-(\textbf{J})
addition. It is immediately obvious that for $l^{th}$ perfect access
pattern, node $h_{ik}$ exercises its two connections to
($a^{m0}_{(2l)}+d_{ik}$) and ($a^{m0}_{(2l+1)}+d_{ik}$).

Let $p$ = [($a^{m0}_{(2l+1)}+d_{ik}$) modulo-(\textbf{J})]
modulo-(\textbf{J/q}), and $\hat{p}$ =
[($a^{m0}_{(2l)}+d_{ik}$) modulo-(\textbf{J})]modulo-(\textbf{J/q}). Then, from theorem \ref{th1}, it
is straightforward to see that the PMUs accessed by
$i^{th}$ one out of (\textbf{J/q}) PPUs of
$k^{th}$ fold for the $l^{th}$ perfect access pattern are
found in appropriate bins of $p^{th}$ and $\hat{p}^{th}$ PMUs.
Table \ref{fold_ex_tab} is organized to explicitly exemplify such
folded mappings.
These PMUs are collocated with the PPUs on the
\textbf{other} side. The algorithm of deriving the pairing is summarized in
algorithm \ref{alg2}.
One can immediately see that while the number of LMUs have
decreased, the size of each PMU has increased proportionally. Hence
this design is a definite case of \textit{linear folding}.

\noindent \uline{\textit{Identical PMU Indices}}

A special case may arise when $p$ = $\hat{p}$, due to the modulo operation,
for a \textit{particular} \textbf{full} perfect access pattern. Then
the data corresponding to two consecutive edges of
\textbf{each} node of the entire non-folded graph get stored in same PMU.
In that case, both the data corresponding to $l^{th}$ perfect access pattern
access are found in the same PMU. The whole architecture
\textbf{still works}, as discussed in section \ref{prob_form_sec}. This
addresses the $3^{rd}$ issue raised in section \ref{fissues_sec}.
Since both data are to be fetched concurrently in a cycle by each
PPU from the same PMU in this perfect pattern, two ports per 2-to-$\hat\rho$ and $\hat\rho$-to-2
switches belonging to one paired, complemented set are used between each
pair of such matched (PPU to PMU mapping)
switches simultaneously. As expected, the concurrent usage of such pair of
ports is itself synchronized across all switches of same type, for both
2-to-$\hat\rho$ and $\hat\rho$-to-2 switches within their respective sets.
Since our interconnect graph is symmetric, exactly the same scheme can be
used to place the data produced by PPUs of the other side.

\subsubsection{Internal Layout of PMUs}
\label{mem_lout_sec}
Now we try to address the $4^{th}$ issue raised in section
\ref{fissues_sec} (and $2^{nd}$ issue of section \ref{missues_sec}
partially). To summarize, this issue relates to finding out
one/both the locations within a PMU, which is read-accessed by a
particular LPU w.r.t. execution of a particular \textbf{full} perfect
pattern. In section \ref{pat_sec}, we pointed out two different ways by
which we can combine the 2-input computations done by a fold. Ideally, the
internal layout of each PMU may simply \textit{follow the
time-order} in which the edges incident on it are scheduled. In such a
case, the address generation unit becomes simply a counter. We do
the layout design with this as objective. The layout is
\textit{briefly} described
only for conceptual clarity, and does not directly result in any design
step. It influences the design of address generation scheme, though, and
hence its value.

This internal layout depends on the design option chosen. In the
following, we explain the internal layout for first design option.
Deriving the layout for second design option on similar lines is
straightforward.

For this option, the first level substructure arises by making `$\gamma$/2'
\textbf{bins} within each PMU, one bin for each of the
$\gamma/2$ \textbf{full} (non-folded, rolled out) perfect access pattern.
A \textit{bin} is defined as a contiguous chunk of memory within
the unit. Whether
for some perfect access pattern, the re-mapped indices of 2 LMUs are same
or different, one can easily prove that the number of bins \textit{remains
constant}. The size of each bin is thus a constant as well, $2\cdot
\mathbf{q}$. Whenever $\gamma$, the degree of each node in
bipartite graph, is odd, the \textbf{last} bin contains only \textbf{q}
real data items, and \textbf{q} items corresponding to storage of dummy
edges. Given the overall size of each PMU, this wastage is
negligible. The bins are arranged in \textit{linear order} with respect to \textbf{full} perfect access patterns.
Hence the address generator simply needs to generate addresses in linear
order in each cycle, whenever read needs to be performed. For
write, the addressing is structured but not linear; see section \ref{add_gen_sec}.
In the execution of a perfect access pattern, each PPU accesses two
memory locations. It may access them either in \textbf{same} PMU, or in
\textbf{different} PMUs.
\begin{itemize}
\item In the former case, \textbf{assume} that $\mathbf{i}^{th}$
one out of (\textbf{J/q}) PPUs of $\mathbf{k}^{th}$ fold (0 $\leq$
\textbf{i} $<$ \textbf{J/q}, 0 $\leq$ \textbf{k} $<$ \textbf{q})
stores \textbf{both} it's data in (some)
$\mathbf{p}^{th}$ PMU (see section \ref{mem_indx_sec} for
calculation of \textbf{p}). If the index of current perfect
pattern being executed is \textbf{l}, then these two data are in
$\mathbf{l}^{th}$ bin of $\mathbf{p}^{th}$ PMU in two consecutive
locations. The \textbf{offset} of these locations from start of the bin is
expectedly, \textbf{2k}
and \textbf{(2k+1)}.
\item In the latter case, \textbf{assume} that $\mathbf{i}^{th}$
one out of (\textbf{J/q}) PPUs of $\mathbf{k}^{th}$ fold (0 $\leq$
\textbf{i} $<$ \textbf{J/q}, 0 $\leq$ \textbf{k} $<$ \textbf{q})
stores exactly \textbf{one} data in (some)
$\mathbf{p}^{th}$ PMU. The possible values of \textbf{p} are fixed
as detailed in section \ref{mem_indx_sec}. If the index of perfect access pattern
is \textbf{l}, then (one of the two) data is placed in $\mathbf{l}^{th}$ bin
of $\mathbf{p}^{th}$ PMU. Since we are folding a perfect
pattern, \textbf{exactly two} edges will have their re-mapped
LMU indices as that of a particular PMU. Hence,
if $\hat{\mathbf{i}}^{th}$ one out of \textbf{J/q} PPUs of
$\mathbf{k}^{th}$ fold \textbf{also} accesses
$\mathbf{p}^{th}$ PMU, and if $\mathbf{i}$ $<$
$\hat{\mathbf{i}}$, then the offset of location for data
corresponding to $\mathbf{i}^{th}$ PPU from start of the bin
is \textbf{2k}, while that of
$\hat{\mathbf{i}}^{th}$ is \textbf{(2k+1)}.
This accounts for address mapping relative to circulant rotation of edges
in the folded graph (see figure \ref{perf_pat}).
\end{itemize}

\begin{figure}[h]
\begin{center}
\includegraphics[scale=0.7]{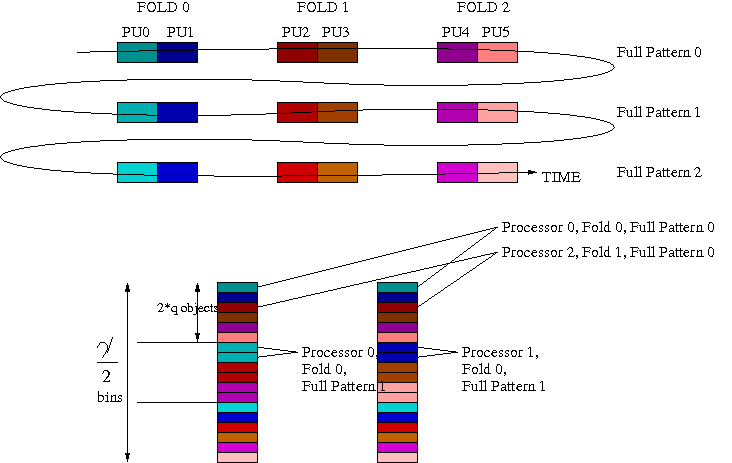}
\end{center}
\caption{Memory Layout for First Perfect Access Pattern Generation Scheme}
\label{fold_mem}
\end{figure}

A \textit{color-coded} version of memory layout
for \textbf{first} design option is shown in figure \ref{fold_mem}.
The parameters of graph in this figure are \textbf{J}=6,
$\mathbf{\gamma}$=6, and \textbf{q}=3. Hence there are \textbf{J/q}=2
PMUs. The set of 3 \textit{similar-colored} boxes in
each column, \textbf{PU*}, represent excitement of all the 6 edges incident
on them at appropriate time, 2-at-a-time. These two edges represent the two
data items consumed by each PU in a cycle. The \textbf{same color} has been
used to depict the location of these two data items in the two
PMUs. Each PMU has $\mathbf{\gamma/2}$= 3 bins, one
corresponding to each perfect access pattern. Each bin has $\mathbf{2\cdot
q\,=\,6}$ data item placeholders. For example, the two data items used
by PU3 during execution of third perfect pattern can be found in $3^{rd}$
bin of both PMUs, in $4^{th}$ location relative to start
of the bin, one each in both PMU. Both these placeholders have
same color as of the box under PU3 for $3^{rd}$ full pattern. Depending on
the perfect access pattern, a particular PPU may store both its
data items in same PMU, or not. This fact can easily be seen to be
dependent on the two indices of destination vertices of the two edges that
are being scheduled as part of that particular perfect access pattern. So,
in this example, for the $2^{nd}$ pattern, each PPU stores both
its data items in same PMU, while it does not for remaining
patterns. It can now be seen that the address generator unit is simply a
counter, the topic that we cover in next section.

Because we schedule binary operations on the PPUs in each
cycle, the PMUs are all \textbf{dual-port} memories.

\noindent \uline{\textit{Layout of Units for Local Access}}

The above layout of PMUs was evolved for read access required by
each computing node. After computing, data corresponding to each edge is
written into local PMU of the computing node. Since the same
PMUs are later accessed by PPUs on the other side of bipartite
graph for input data, the data written into these local units needs to be
organized again in the same form, as discussed above. In fact, the address
generation scheme for writing also remains same, as that of the read
accesses that follow. This addresses the $5^{th}$ and $6^{th}$ issues
raised in section \ref{fissues_sec}. To summarize, the former
issue relates to finding the location of the local memory of a
particular PPU in which data corresponding to an edge has to be written
into, while the latter relates to knowing the machine cycle in which
writing has to be performed.

\subsubsection{Address Generation}
\label{add_gen_sec}

Here, we \textit{first} address the \textit{refined} $2^{nd}$ issue raised in section
\ref{missues_sec}: if the PPU $h^{m0}$
working on some binary operation accesses the $\mathbf{m}^{th}$ PMU, then in which cycle does it access it, and at which location (local
address)? To simplify generation
algorithm, we take $h^{m0}$ as $h_0$, and \textbf{m} as $\mathbf{t}^{th}$
PMU connected to it, as discussed in section \ref{edge_map_sec}.

As such, the address generation requirements are apparent from the memory layout
and flow of time, as depicted in figure \ref{fold_mem}. Since we can
combine balanced patterns for a fold in two different ways to form a
perfect sequence, the requirements also correspondingly differ. For
illustration as well as continuation, we take the \textit{first
design option} again.
We now calculate the schedule for $\mathbf{t}^{th}$ edge of any
node, which is shift-replica of $\mathbf{t}^{th}$ edge of $h_0$. Details of
this replication were discussed in section \ref{mem_indx_sec} earlier.

\begin{lem}
\label{lem1}
For the first design option, the $\mathbf{t}^{th}$ data associated with LPU $h_{ik}$ is accessed from some PMU's some
location (computable from section \ref{mem_lout_sec}) in cycle number
{\large $\left(q\cdot \lfloor\frac{t}{2}\rfloor + k + 1\right)\cdot
\mathbf{T}$}, where \textbf{T} is the number of machine cycles
taken for completion of computation by each node.
\end{lem}

\begin{proof}
Each PPU computes on behalf of \textbf{q} overlaid LPUs in first design option, per perfect
pattern. Further, before arriving at the right (current) perfect access pattern
in which $\mathbf{t}^{th}$ data is consumed, {\large $\lfloor\frac{t}{2}\rfloor$}
\textbf{full} perfect access patterns must have completed execution. This
is because by definition, $l^{th}$ perfect pattern is one that excites $(2
\star l)^{th}$ edge of $h^{m0}$; see section
\ref{mem_indx_sec}. Due to overlay, 
LPU $h_{ik}$ gets scheduled during the current
perfect access pattern only in cycle number (\textbf{k+1}),
counted from the \textit{beginning of the current} perfect access pattern. These two
components add up to give the cycle number required.
\end{proof}

It is straightforward to further note that the \textbf{J/q}
circulantly shifted replicas of $\mathbf{t}^{th}$ edge of $h_{ik}$, within the
\textbf{same} fold, also get scheduled in the same cycle. By
varying the values of \textbf{t} and \textbf{k}, we
can cover schedule for all the edges of all nodes, i.e. the
\textit{complete schedule}. Knowing the two
locations per cycle in each PMU that the schedule uses, the address
generation counters of various PMUs can be \textit{synchronized}.
The algorithm for address generation is
summarized in algorithm \ref{alg3}.

\begin{algorithm}
\caption{Address Generation for First Design Option}
\label{alg3}
\begin{algorithmic}[0]
        \ForAll{PMUs $a_i$, $0 \leq i < \mathbf{J/q}$, connected to $h_0$ }
    \State Find the position of edge, \textbf{t}, between $h_0$ and $a_i$, by doing
    a side-to-side scan of edges connected to $h_0$ \\
    \Comment Assume that each node computation takes \textbf{T} machine
    cycles ~~~~~~~~~~~~~~~~~~~~~~~~~~~~~~~~~~~~~~~~~~~~~~~~~~~~
    \State LPU $h_{ik}$, overlaid on some PPU,
    accesses some location of some PMU\\
    ~~~~~  (computable from section
    \ref{mem_lout_sec}) in cycle number {\large $\left(q\cdot
    \lfloor\frac{r}{2}\rfloor + k + 1\right)\cdot \mathbf{T}$} \textbf{onwards}
        \State The shift replicas of this edge within same,
        $k^{th}$ fold, get scheduled in same cycle, too
\EndFor
\end{algorithmic}
\end{algorithm}

Continuing the example graph of
table \ref{h_p_ex_tab}, let \textbf{t} = 5, so that the $5^{th}$ edge of $h_0$ ends on
\textbf{p5}. Assuming the earlier fold factor \textbf{q} as 3,
\textbf{h0} is in first fold of the graph. Hence the
$5^{th}$ edge of \textbf{h0} is
scheduled in {\large $\left(3\cdot \lfloor\frac{5}{2}\rfloor +
0 + 1\right)\cdot \mathbf{T}$} = $
7\cdot \mathbf{T})^{th}$
clock cycle.

We also state without proof, another address generation scheme.
\begin{lem}
\label{lem2}
For the second design option, the $\mathbf{t}^{th}$ data associated with
LPU $h_{ik}$ is accessed from
some PMU's some location (computable from sections
\ref{mem_indx_sec} and \ref{mem_lout_sec}) in cycle number {\large
$\left(\frac{\gamma}{2}\cdot\mathbf{k} +
\lceil\frac{r}{2}\rceil\right)\cdot \mathbf{T}$}.
\end{lem}

Each PMU is a true dual-port memory, and hence each port requires
a separate address generator. If we stick to the
convention defined next, it is easy to verify that \textit{both the address
generators will be a counter}. Assume that the execution of
\textit{next}
perfect access pattern needs to be scheduled at each port now. Each PPU accesses two memory locations. For the next pattern, it may access
them either in \textbf{same} PMU, or in \textbf{different} 
units.
\begin{itemize}
\item In the former case, exactly one PPU per fold will store
both its data items of this pattern in the particular PMU. Then,
in the relevant machine cycle, let the defined convention be that the
\textit{first port} read/write the data item at offset \textbf{2k} from the
beginning of the bin corresponding to this pattern.
By similar convention, in the same cycle,  \textit{second port} reads/writes the data item at offset
\textbf{2k+1} from the beginning of the bin corresponding to this pattern.
Here, \textbf{k} is the index of the fold that is currently being
scheduled.
\item In the latter case, exactly two PPUs per
$\mathbf{k}^{th}$ fold read/write one data item each into the
PMU in question.
Let the re-mapped indices of these PPUs (after folding) be
$\mathbf{i}$ and $\hat{\mathbf{i}}$. Also, without loss of generality, let
$\mathbf{i}$ $<$ $\hat{\mathbf{i}}$. Then, in the relevant machine cycle,
let the defined convention be that the \textit{first port} read/write the
data item at offset \textbf{2k} from the beginning of the bin corresponding
to this pattern, which is exchanged with PPU $\mathbf{i}$.
By similar convention, in the same cycle, 
the \textit{second port} then reads/writes the data
item at offset \textbf{2k+1} from the beginning of the bin corresponding to
this pattern, which is exchanged with PPU $\hat{\mathbf{i}}$.
\end{itemize}

\noindent \uline{\textit{Write Address Generation and Multiplexing}}

We now address the related address generation issue pointed out in section
\ref{missues_sec}: in write-back phase to local memory by a PPU, in what sequence of locations must the output data generated in
successive clock cycles be stored? We had hinted that the order in which
PPUs of other side/type will access this generated datum as
their input input, dictates the sequence of locations in the local memory.

We start by observing that \uline{in absence of folding}, the data must be
written in reverse (linear) order of locations into local memory. From
previous section, the read order of a PMU was found to start from
$0^{th}$ location, and increase in a step of 1 till the last location,
which we term as \textit{forward linear} order. The write order, which is
reverse of this, is hence termed as \textit{reverse linear} order. This is easy to
prove using circulance property of the perfect matchings that form the
each perfect pattern, which in turn combine to form the perfect access
sequence. Take two successive edges incident on a node having
index $s$, on one side of
the graph, and let $d$ and $\hat{d}$ be the indices of end points
of these edges (on other side of graph) such that
$\hat{d}$ $>$ $d$ without loss of generality. These two edges are part of
two different perfect matchings. When we look at e.g. node $d$ and observe
the perfect access pattern to which these two edges belong, one can see
that the node $s$ contributes one of these edges incident on it, \textit{plus}
an edge that is part of the perfect matching to which the other edge
belongs, to the (same) perfect access pattern. Let the other end of this
different edge be a node having index $\hat{s}$. If $\hat{d}$ $>$ $d$, it is straightforward to
prove that $s$ $>$ $\hat{s}$. Hence for read order to be forward linear,
the write order must be reverse linear, \uline{in absence of folding}. For
the example folding of graph of table \ref{h_p_ex_tab}, one can see this
order in table \ref{write_gen_tab}. The table tabulates the
data input sequence of \textit{point nodes},
as generated in certain order by various
\textit{hyperplane nodes}. In the table, \textbf{A}-\textbf{O} are
(15) hyperplane labels, and for each hyperplane, e.g.,
\textbf{A},  \textbf{A0} represents $0^{th}$ edge data, out of 8 edge
datum\footnote{7 real + 1 for dummy edge for last perfect access
pattern},
generated by hyperplane \textbf{A}. Further, the numbering \textbf{0-6} has
been done based on perfect access pattern-based grouping of edges, that are
incident on the consumer (point) nodes. Thus, from table
\ref{h_p_ex_tab} or Fig. \ref{15_pg_fig}, \textbf{A0} is the data that is
read from hyperplane \textbf{0} by point \textbf{0}, \textbf{A5} is the
data that is read from hyperplane \textbf{0} by point \textbf{8}, and so
on.

\begin{table}[h]
\caption{Sequence of Data Items Consumed by Point Nodes of Graph in Table
\ref{h_p_ex_tab}}
\label{write_gen_tab}
\centering
{\normalsize
\begin{tabular}[!h]{|c|c|c|c|c|c|c|c|}
\hline
\textbf{Point Index} & \multicolumn{7}{|c|}{Sequence of Data Item
Output} \\ \hline \hline
0 & A0 & F6 & H5 & K4 & L3 & N2 & O1 \\ \hline
1 & B0 & G6 & I5 & L4 & M3 & O2 & A1 \\ \hline
2 & C0 & H6 & J5 & M4 & N3 & A2 & B1 \\ \hline
3 & D0 & I6 & K5 & N4 & O3 & B2 & C1 \\ \hline
4 & E0 & J6 & L5 & O4 & A3 & C2 & D1 \\ \hline
5 & F0 & K6 & M5 & A4 & B3 & D2 & E1 \\ \hline
6 & G0 & L6 & N5 & B4 & C3 & E2 & F1 \\ \hline
7 & H0 & M6 & O5 & C4 & D3 & F2 & G1 \\ \hline
8 & I0 & N6 & A5 & D4 & E3 & G2 & H1 \\ \hline
9 & J0 & O6 & B5 & E4 & F3 & H2 & I1 \\ \hline
10 & K0 & A6 & C5 & F4 & G3 & I2 & J1 \\ \hline
11 & L0 & B6 & D5 & G4 & H3 & J2 & K1 \\ \hline
12 & M0 & C6 & E5 & H4 & I3 & K2 & L1 \\ \hline
13 & N0 & D6 & F5 & I4 & J3 & L2 & M1 \\ \hline
14 & 00 & E6 & G5 & J4 & K3 & M2 & N1 \\ \hline

\end{tabular}
}
\end{table}

\uline{In presence of folding}, the write order has to \textit{factor in
interleaving} of data, as is done by the overlaid (point) nodes. 
For example, recall that hyperplanes \textbf{A}, \textbf{F} and \textbf{K} are
overlaid, and so on. Data corresponding to edge \textbf{A0} is consumed by
point node \textbf{0}, which belongs to \textbf{first} fold of point nodes
(\textbf{i} = 1), during \textbf{its} execution of \textit{first} perfect
access pattern (\textbf{j} = 1). Similarly,
data \textbf{A4} is consumed by point node 5, which belongs to
\textbf{second} fold of point nodes (\textbf{i} = 2), during \textbf{its} execution of
\textit{second} perfect access pattern (\textbf{j} = 2). Since \textbf{A0}
is to the \textit{left of} \textbf{F6}, when looking from point node
\textbf{0}, as in Fig.
\ref{15_pg_fig}, \textbf{A0} is stored in ((i-1)$\times$6)+(j-1)*2+0 =
((1-1)$\times$6)+(1-1)*2+0 = $\mathbf{0}^{th}$ location. Similarly,
\textbf{F6} is stored in ((i-1)$\times$6)+(j-1)*2+1 =
((1-1)$\times$6)+(1-1)*2+0 = $\mathbf{1}^{st}$ location, and \textbf{A4}
in ((i-1)$\times$6)+(j-1)*2+1 = ((2-1)$\times$6)+(2-1)*2+1 =
$\mathbf{9}^{th}$ location. This way, \textit{algorithmically}, the entire
folded write order can be generated.

Such a write-back address sequence can generally be implemented using an
LUT. A multiplexer is also generally needed to choose between read and
write address generator's outputs, to be interfaced with PMU's
address inputs, in a particular clock cycle.

\noindent \uline{\textit{Implementing the Generator}} \\
There are two ways by which PPUs can access operands stored in
PMUs in a particular cycle. In the first way, the PPUs
themselves calculate/generate and place the address/location using an extra
bus, for a memory access. This is a standard practice in
von Neumann
architectures. Since there is a deterministic structure in access order, it
is possible to do the other way round. One can alternatively build and
embed an address generator within the PMU (alongside its controller), which places two data
objects on the two ports (or alternatively, allows two data objects to be
stored at two locations), given the cycle number. For each PMU, we need one address generation unit, in either case.

\subsection{Derivation of Complete Schedule}
\label{comp_sched_sec}
With the individual issues related to complete schedule derivation for a
folded PG-based system addressed in previous section, we now describe how
the entire computational schedule, \textbf{without pipelining}, can be
arrived at. It is easy to understand this schedule by looking at the
detailed structure of the system, as in figure \ref{full_arch_fig}. We
assume that LPUs of first type take $P_1$ units
of time, and of the second time take $P_2$ units of time, for
their computation. A PPU
is an overlay of \textbf{q} LPUs, and hence the two types of
PPUs take $\mathbf{q}\cdot P_1$ and $\mathbf{q}\cdot P_2$
units of time to compute, respectively. The required \textit{expansion} of
this interval of computation, based on the design option chosen as per
section \ref{pat_sec}, can be easily generated from original schedule
interval of one representative out of the grouped PPUs.

After e.g. first type of PPUs finish computation, the output
will need to be stored into local PMUs. There are $\gamma$ edges
per node, overlaid \textbf{q} times. Accounting for dummy edges whenever
$\gamma$ is odd, $\lceil \frac{\gamma}{2} \rceil \cdot \mathbf{q}$ units of
time will be taken by each PPU to write back all its output
data into local PMU. The schedule for this interval is simply a
counter that drives the write-back location generation logic, and hence can
be easily extended by a factor of \textbf{q}.

After local storage, the new data will be required by the PPUs
of other side. This requires participation of both 2-to-$\hat\rho$ and
$\hat\rho$-to-2 switches in \textit{almost lockstep} fashion.
More specifically, to allow the data to be read from one
end from a PMU, and passed across the other end of the
interconnect to a PPU, switches of both types in each
of the two sets are active in same set/interval of machine cycles, except one
cycle each at \textit{either} end of the interval. This minimal staggering
is because the system being a \textit{completely
synchronous system}, $\hat{\rho}$-to-2 switches can only be activated one cycle
later, after 2-to-$\hat{\rho}$ switches have put the data on the interconnect
wires. The cycle interval in which switches in each set are
active, starts at a cycle number computable
from section \ref{add_gen_sec}, and lasts {\large $\mathbf{T}\cdot
\frac{\gamma}{2}$} cycles. Here, \textbf{T} is equal to either
$P_1$ or $P_2$, depending on which PPUs require the data. The
data is read, for one cycle, only \textbf{every} \textbf{T}
cycles. Hence switches are \uline{periodically enabled} every
\textbf{T} cycles.

The above schedule is symmetric, and hence with appropriate change in the
set of signals, can be used to derive the other half of
schedule, in
which other sets of PPUs, local PMUs, and
2-to-$\hat\rho$ and $\hat\rho$-to-2 switches are involved.

\subsubsection{Complete Schedule with Pipelining}
\label{pipeline_sec}

Pipelining the above system leads to saving of clock cycles to some
extent, and corresponding recovery of throughput. In a
\textit{partially or fully structural model} of a VLSI
system that is composed of \textit{component hierarchies}, pipelining can
be tried out between every two components that are \textbf{adjacent to each
other} in the data flow, and belong to same level, at every level of
component hierarchy. For our intended system, pipelining
can be performed at \textbf{three}
levels. It can be tried at the graph level, by trying to
pipeline computation done by one type of PPUs, with the other
type of PPUs. It can also be tried at the
high-level architecture level, as in figure
\ref{full_arch_fig}, and finally at micro-architecture level,
i.e. computation done by each node. In the latter case, each node can
consume 2 inputs (1 at each port) every clock cycle, and hence value of
\textbf{T} becomes 1 for the sake on periodic input consumption.
In the former case, one can, for example, pipeline
the write-back phase of a PPU. As soon as a PPU is
ready with some data that can be output, it starts storing it in its local
PMU in a \textit{pipelined fashion}. A prototype design that we did
using this methodology uses pipelining wherever feasible. Doing such
pipelining will shrink the simple folded schedule discussed earlier.
However, with appropriate guidelines, the above shrinking can also be
automated. The (positive) impact of these two levels of pipelining
on throughput depends on the time taken by each PPU,
\textbf{T}, which varies across systems being modeled. Hence the
improvement figure is not generalizable.

Finally, for pipelining at the graph level, the second
design option discussed in section \ref{pat_sec} opens up an
avenue to do coarse-grained pipelining of the system. Recall that in this
design option, we may first sequentially schedule all $\gamma/2$ 2-input
computations done by each PPU in one fold only, which cover up
the complete computations of \textbf{J/q} nodes in the non-folded version.
In default mode, the system scheduler waits for $\mathbf{(q-1)}$ more
rounds of such computations to cover remaining nodes of one side of the
unfolded graph, and then schedules the communication of the results of
entire one side computation to the PMUs belonging to the
PPUs on
other side of the graph. Instead, we can start communication as soon as
$\mathbf{J/q}$ computations over PPUs of one fold is over. \textbf{In parallel}, we can also
start doing computation for next lot of \textbf{J/q}
PPUs.

To characterize the impact of this level of pipelining
on throughput, we assume that
2-input computations by each PPU happen in a
single cycle. Further, due to dual-port memory assumption, and no
write/write conflict while writing into PMUs (see section
\ref{mem_lout_sec}), one can assume that 2 data get stored in a memory
unit per cycle. However, there may be additional communication latency
due to e.g. passage of data through switches, before it arrives at the
port of memory units. Assume this constant latency to be $\Delta$ cycles.
Then, it is easy to see that each half-iteration (input of data, computation
and communication of resultant data) over all folds takes
{\large $\left(\frac{\gamma}{2}\times \mathbf{q} + 2\Delta\right)\cdot
\mathbf{T}$} cycles
\textit{optimally}. This is almost a two-fold improvement over a
non-pipelined design, where a half iteration would have taken
{\large $\left(\frac{\gamma}{2}\times \mathbf{q}\right)\cdot \mathbf{T}$} cycles. The cost of
$\Delta$ can be amortized in the case of big-sized problems (higher
$\gamma$), as is practically always the case.

\subsection{Putting it all Together: Summary of Design Methodology}
We start the usage of this methodology by accepting an annotated
PG bipartite graph as input specification, in which the nodes are
annotated with their untimed behavior. The graph is
parameterized in terms of order \textbf{J} and (regular) degree $\gamma$. If
not pre-sorted, then the bi-adjacency matrix of the graph is first sorted
so that the \textit{circulant symmetry} inherent in PG bipartite graphs
becomes explicit. If \textbf{J} is a prime number, we first expand
the graph to non-prime order, as in section \ref{prime_sec}. The choice of
number of nodes, $\alpha$, to be added on each side of the graph can be
influenced by two factors. One is the factorizability of
$(\mathbf{J}+\alpha)$, and the other is whether for some value of $\alpha$,
equation \ref{eqn_grow} becomes an equality. In such case, the expanded
degree of each node is lesser. We then
calculate all possible factors \textbf{q} of
\textbf{J}. We finally select one of these factors based on various
judgements. One of the possible reasons could be if the modulo operation of
end point of two edges leads to the same index or not. Another reason could
be the overall area budget (for example, as approximated using gate count).
We then instantiate \textbf{J/q} PPUs and PMUs, as well as \textbf{J/q} 2-to-$\hat{\rho}$ and $\hat{\rho}$-to-2
switches to interface them. This set of components correspond to
one side of the bipartite graph, and hence is further \textbf{duplicated}
to implement the other side of the bipartite graph as well. The
internal micro-architecture of PPUs is then suitably modified
to handle folding, as per section \ref{micro_mod_sec}. Local
interconnect is added between each of the two ports of each of the
$\hat{\rho}$-to-2 switch, and a port of its local PPU. Local
interconnect is also added between each of the two ports of each of the
2-to-$\hat{\rho}$ switch, and a port of its local PMU. Two instances of global interconnects, one each
between the 2-to-$\hat{\rho}$ and $\hat{\rho}$-to-2 switches of
\textit{opposite} sides, are designed using guidelines in section
\ref{mem_indx_sec}. We then generate the \textit{folded} perfect
access patterns for communication over these global interconnect instances,
as per algorithm in section \ref{pat_sec}. If any initialization data is
to be provided to any type of LPUs, it is provided in a
multiplexed way to the overlaid PPUs, at the beginning of the
computation. Similarly, any output data from LPUs of one type
is to be physically obtained by demultiplexing the output of corresponding
overlaid PPUs. At this point, the control path and the
timing of the system are evolved. The invocation (start) of this sequence
signifies flow of data inputs for PPUs on one side of graph,
from PMUs located on other side of graph.
Accordingly, partial computations can be done on these
PPUs, as soon as some subset of data arrives. At the end of invocation of
one complete perfect sequence, one side of graph is through with its
parallel computation. Another invocation of perfect sequences communicates
the resultant data into the local memory of PPUs on other side
of the graph. These PPUs can then again start acting
immediately on this recent data. If the computation is iterative, the same
sequence repeats. The address generation of various PMUs,
(whose layout is described in section \ref{mem_lout_sec}) whenever
a perfect sequence is active, is governed by the algorithm in section
\ref{add_gen_sec}. The generation of selection signals for various
switches (described in section \ref{switch_sec}) is governed by
derivations in section \ref{edge_map_sec} and \ref{mem_indx_sec}.
The derivation of overall schedule is finally done, as discussed
in section \ref{comp_sched_sec}.

\section{Models, Refinement and Design Space Exploration}
\label{mod_ref_sec}
As introduced so far in this paper, we use five successive levels of
abstraction for
models, and correspondingly four refinements in our methodology. We now
show the correspondence of this methodology to general synthesis-based
communication architecture design methodologies, both generic and specific.
Such correspondence was found out post-specification of this methodology,
reinforcing our belief that practical, useful design flows can be
implemented for this methodology.

\subsection{Model Abstraction Levels in Generic SoC Design}
In generic SoC design, following models are used at various levels of
abstraction \cite{ocn_book}, \cite{comm_soc}.
\begin{description}
\item[Functional Model] is generally a task/process graph model, capturing
just the functionality of the system. 
\item[Architecture-level Model] is created by refinement of functional
models. They introduce various hardware blocks/components,
hardware/software partition (if any), their behavior and abstract channels
for inter-communication. \\
Such models belong to the category of
transaction-level models supported by various system-level
languages, which model communication events between modules over such
channels, and their causality etc. \cite{gajski_tlm_pap}.
\item[Communication-level Model] is created by refinement of e.g.
transaction-level model, and describes the system communication
infrastructure in more detail, many a times to the
cycle-accurate level of granularity, or to an approximation of
it otherwise \cite{gajski_tlm_pap}.
Most amount of design space exploration for communication architecture
design happens at this level. The computation details are generally not
refined, while refining a transaction-level model.
\item[Implementation-level Model] is generated by refining
communication-level model, and captures details of \textbf{all} the
components of computation and communication subsystems at the signal and
cycle-accurate level of detail. They are typically used for
detailed system verification and even more accurate analysis.
\end{description}

We now explain the correspondence of abstraction levels. In our design
methodology, the
starting graph is a Tanner graph additionally annotated with each node's
\textit{untimed behavior}, i.e. the functionality.
This suffices to be the
\textit{functional model} for the
intended system. The \uline{first} level of refinement to this model,
defined in section \ref{non_simd_pap_sec}, adds some details (such as
barrier sync requirement) to this model, \textbf{specific} to the class of
applications this methodology targets. This refinement is itself
\textit{optional}, and leads to a functional model only. The \uline{second}
level of refinement takes the functional model to \textit{architecture
level}, and
is explained in section \ref{sys_arch}. Real PPUs and
PMUs are  assigned and cross-connected at this level. These
connections represent channels that carry the \textit{uniform}
communication traffic as per \textit{Flooding Schedule}. Main part of
design space exploration is carried out next, as discussed in next section.
This \uline{third} level of refinement transforms the
set of channels in architecture model to a
cycle-accurate communication model, in form of the generated
\textit{folded communication schedule}, as in section \ref{pat_sec}.
The specification of computation is also refined to introduce
timing, as per section \ref{micro_mod_sec}. The overall system is thus
\textit{approximately-timed}, as defined in \cite{gajski_tlm_pap}. There
are two design options to be explored at this level; see section
\ref{pat_sec}. Finally, the \uline{fourth} level of refinement takes this
schedule to implementation-level model, which corresponds to
generation of RTL for all components of the communication
subsystem (switches, address generators etc). From this point onwards,
successive refinement to more detailed models based on some standard
RTL-based design flow is done to complete the design.

As one can observe, we do not need a high-level model more complex than
an annotated
bipartite graph to start with, unlike e.g. Kahn Process Networks as
starting model in COSY methodology \cite{cosy_methodology}. Similarly, we
do not need standard intermediate level models such as VCI
models, again in COSY methodology.

\subsection{Similarity to Levels in SpecC Design Methodology}
SpecC language was created by Gajski et al in the backdrop of evolving a
system-level, platform-based design methodology \cite{spec_c_methodology}.
It uses four model abstractions: specification, architecture, communication, and
implementation. The first, specification model level, is defined to
capture the functionality of the system using sequential or concurrent
behaviors that communicate via global variables or abstract channels. It is
similar to functional model mentioned by us in
previous section, and hence a Tanner graph suffices to be again
called a specification model. The architecture, communication and
implementation levels have same meaning as in previous section, but
\textit{in context of SpecC language constructs}. Without going into more
details here, we have found that our models and refinements again
correspond closely to models and refinements defined in SpecC-based
methodology. As in our case, the implementation model, as an RTL model, is
passed on to some standard design flow.

\subsection{Design Space Exploration}

As discussed in beginning of this paper, this folding scheme can
\textbf{also} be viewed as one of evolving custom communication architecture.
Since we use a \textbf{custom} communication architecture, once the custom
architecture is fixed, the next step is usually to perform an
\uline{exploration} phase of the \uline{design space} \cite{ocn_book}.
On-chip communication architecture \textit{design space} is generally a
union of topology and (communication) protocol parameter spaces, and
\textit{exploration} is needed to determine the topology and protocol
parameters that can best meet the design goals. The protocol can be a set
of communication mechanisms working together (e.g.  routing, flow control, switch
arbitration etc. in case of a network-on-chip). The protocol parameters
need to be decided to satisfy various application constraints. These
constraints generally relate to performance, power, area, reliability etc.

It is easy to recognize from the earlier summary of methodology, that the
choice of fold factor, \textbf{q}, impacts at least the throughput and area
figures. As such, \textbf{q} is a parameter that is required to specify the
topology (number of vertices per fold, and hence number of point-to-point
connections needed). Also, at times when the number of nodes on one side of
the bipartite graph, \textbf{J}, is prime, we need to add a variable number
of nodes, $\mathbf{\alpha}$ to make the graph size factorizable. Hence a
limited amount of topology exploration, by varying \textbf{q} and
$\mathbf{\alpha}$, is needed, as already hinted in the summary earlier.
Protocol exploration is not needed in its full detail,
since the choice of algorithms driving various components is already
fixed (detailed throughout section \ref{methodology_sec}), and is
\textbf{optimum} for each component (e.g. linear addressing for PMUs) due to various customizations. The
\textit{lone} important protocol parameter to be decided is the wire
switching frequency, which can be fixed without any algorithm-level
explorations.

If one looks at throughput constraint, then it is governed by
both switching frequency as well as the value of \textbf{q} (\textbf{q}
participates in throughput-area tradeoff, as pointed earlier). If one looks
at energy consumption, then switching frequency alone governs the energy
consumption, and not the value of \textbf{q}. These constraints provide the
desired switching frequency, generally as an interval (throughput constraint
providing lower bound and power constraint providing upper bound). This
also stems from the fact that power and performance generally trade off in
system design. The \textit{actual} switching frequency can only be determined during physical
design phase, based on placement-and-routing information. Since we suppose
that beyond RTL generation, a standard synthesis flow will take over the
remaining system design, in the best case, a high-level floorplanner
\cite{fp_bus_syn} can be integrated with high-level synthesis tool in the
standard design flow part. Integrating these two will logically reduce the
number of iterations needed to fix the frequency. However, it can then take
extra efforts to implement a \textit{feedback loop} across two flows (one
custom and one standard), in order to
explore around the switching frequency. With or without such feedback loop,
the design space with upto two variables, becomes
limited, and can be \uline{explored in polynomial time}. This is unlike other
explorations such as synthesis of bus-based architectures, whose
exploration is generally NP-hard. In those cases, one has to further choose from
various categories of synthesis techniques (simulation-based,
heuristic-based etc), and the exploration time is also higher.

\section{Addressing Scalability}
\label{scale_sec}
As pointed out earlier, this methodology can handle certain scalability
issues. This implies that a new folded system
architectures be designed to handle higher 
input block sizes. Changing the value of \textbf{J} means that the set of all possible
factors (q) of \textbf{J} also change. However, usage of
PG implies that many components such as
individual PPUs, address generation units can be re-used,
with \uline{very limited modifications}. The modifications are in the
contents of LUTs, if any component uses them, and not in the behavior of
the component, such as linearity of address generator.
Similarly, the PMU size increases, though
the internal structure remains same. The switches need to be redesigned,
though.

\section{Advantages of Static Interconnect}
\label{stat_adv_sec}
In this section, we quantify the $4^{th}$ advantage listed in beginning of
this paper. As the computation moves from one fold to another, in our
scheme, same 2-to-$\hat\rho$ and $\hat\rho$-to-2 switches can be used
across the \textbf{q} folds, due to perfect overlay. Same is not true in
case of any other folding. Hence we will either need different multiplexers
and demultiplexers to handle data distribution in each fold, or a single
big multiplexer and demultiplexer which is a union of all these. Also, a
specific control signal will need to be added, which will specify
computations for which particular fold is being carried out. It will be
used to select the corresponding multiplexer and demultiplexer. In the
\textbf{worst case}, in some other folding, there will be upto (\textbf{q -1})
more multiplexers and demultiplexers, one more internal control signal, and
of course upto \textbf{q} times more used wiring resources, since connections
are not getting re-used. Hence our folding scheme offers a lot of resource
saving, and some degree of latency saving.
\section{Prototyping and Evaluation}
\label{exp_sec}
\subsection{Proof of Concept}
For proof of concept, an iterative decoder having a Tanner graph
representation that of the PG bipartite graph example tabulated in table
\ref{h_p_ex_tab} was prototyped in behavioral VHDL. The prototype
has been described in \cite{ldpc_pap}. To recall, the example
has 15 point and 15 hyperplane nodes, each with a degree of 7, in the
bipartite graph. The decoding algorithm employed by the decoder is the
hard-decision bit-flipping algorithm \cite{guilloud}. All the refinements,
and design space exploration was done manually. A fold factor of 3 was used
to fold the bipartite graph, thus requiring (5+5) PPUs and
(5+5) PMUs for implementation, \textbf{plus} (5+5) 2-to-5 and (5+5)
5-to-2 switches. The interconnect between ports of switches of opposite
side was based on guidelines discussed in section \ref{switch_sec}. 
The folded graph schedule was already worked out in table
\ref{fold_ex_tab}. First design option was used to combine perfect access
patterns into perfect access sequences. The edges of various folds were
indeed found to overlay \textbf{perfectly}, following theorem
\ref{th2}. Since the node degree is odd (7), a dummy edge was needed to be
added to each node as expected, and each node would ignore the value
arriving on dummy input during its computation. The micro-architecture of
all nodes was changed to create 3 copies each storage element, since in
bit-flipping algorithm for decoding, all nodes have at least one
computation that consult all inputs (counting all bits or XORing all
bits). 4 LUTs were used to store the port selection schedule to
drive 2 sets of 2-to-5
switches, and 2 sets of 5-to-2 switches. The centralized control path was
implemented using the concept of microcode sequencing. Each
iteration of decoder takes 63 clock cycles, while the unfolded version
takes 35 cycles. Since we implemented two levels (out of three levels
suggested in section \ref{pipeline_sec}) in this design, the throughput
\textit{reduction factor} lessens from being (\textbf{q} =) 3 to
$\frac{63}{35}$, i.e. just 1.8.

\begin{table}[h]
\caption{Parameterized Model of Prototyped System}
\label{prot_sys_tab}
\centering
\begin{tabular}[!h]{|c|c|}
\hline \hline
Order of PG Bipartite Graph & 15 nodes on each side \\ \hline
Degree of each node & 7 \\ \hline
Fold Factor & 3 \\ \hline
Additional Nodes added for non-primality & 0 \\ \hline
Number of PMUs accessed by each PPU ($\rho$) & 5 \\
\hline
Number of PPUs accessed by each PMU ($\rho$) & 5 \\
\hline
No. of output ports ($\hat\rho$) of 2-to-$\hat\rho$ switches & 5 \\ \hline
No. of input ports ($\hat\rho$) of $\hat\rho$-to-2 switches & Same as above \\ \hline
Dummy Edge used in scheduling & Yes \\ \hline
Size of each LMU & 24 data units \\ \hline
Address generation LUTs used & 4 \\ \hline
Computation time for each PPU & 12 clock cycles \\ \hline
Schedule length for 1 iteration & 63 clock cycles \\ \hline \hline

\end{tabular}
\end{table}

The above design methodology was also employed to design a
specific high-performance soft-decision \cite{fossorier} decoder
for a class of codes called LDPC codes. The design has been patented
\cite{ldpc_foldpat}. A detailed C-language simulator was also
developed to verify the entire schedule. Table \ref{fold_ex_tab} was
generated using this simulator. A front-end to generate
per-cycle schedule in form of figures, to \textit{visually} verify various
properties of folding, was also implemented. An animation using such
component schedule figures, depicting the overall schedule,
which was generated by this front-end, can be found in
\cite{fold2_techrep}. All the programs are available from authors
on request.

Similar to real employment of this scheme, the alternative folding
scheme (discussed in \cite{expanders}), was employed to design a DVD-R
decoder using alternative, novel class of error-correction codes developed
by us. The design has been applied for patent as well. For this
decoder system, (31, 25, 7) Reed-Solomon codes were chosen as
subcodes, and (63 point, 63 hyperplane) bipartite graph from {\normalsize
$\mathbb{P}(5,\mathbb{GF}(2))$} was chosen as the \textit{expander graph}.
The overall expander code was thus (1953, 1197, 761)-code. A
fold factor
of 9 was used for the above expander graph to do the detailed design. The
design was implemented on a Xilinx Virtex 5 LX110T FPGA \cite{overview_manual}.

\subsection{Synthesis Tool Prototyping}
To showcase the proof of concept, we have developed a
synthesis tool in C++ language, that aims at implementing a semi-automated
synthesis tool for the methodology. The synthesis tool was designed to emit VHDL mixed behavioral/structural
model. That is, components have a behavioral model, but
they are instantiated structurally wherever used.
The tool software is available on request.

In this software, the first three refinements, which are mainly
data-centric, revolve around populating various data structures.
Implementation of the last stage of refinement relates to emitting the RTL
model, and forms the \textbf{bulk} of the software. To implement this
stage, we started by \textbf{parameterizing} the (envisaged) RTL
specification of the system that is the result of the refinements.
As discussed earlier throughout the chapter, all nodes have a behavioral
template (e.g. the switch). For a given system specification, we
instantiate each system component with appropriate values for the generic
parameters, before integrating them and imposing them with a global
schedule.

To be able to integrate, the signal names and types used at component
interface have been made compatible. The signal name compatibility in the
entity description, and the component instantiation follows the default
rule in VHDL: they are same. Most of the other \textit{default} entity
binding rules of VHDL are also followed.

After parameterization, we proceeded with identifying those portions of
behavioral model of each component, that is affected by change in one or
more parameters. On a case-by-case (entity-by-entity) basis, we devise
small algorithms to generate and emit such portions, given specific values
to the parameters that \textit{affect it}. In the software,
by placing such
generated portions in right place w.r.t. those portions of the hardware
model, that are \textit{unaffected} from any change in any parameter, we
generate the entire behavioral model instance from the
template of that component. We now give an example of
such a generation in our software, for one component, the \textbf{memory unit}, to demonstrate
the generation strategy.

\subsubsection{Memory Unit Generation}
The Memory Unit entity integrates two local address generators, an address
mux and a \textit{dual-port} memory element. One intended behavior of this
entity is to store the computation output by the local processing units
appropriately. The other intended behavior of this entity is to provide/get
read for the computation input by processing units on the other side of the
PG graph. Since we use symmetric graphs, a single memory unit template, and
similar instantiation serves the purpose for memory units that are used in
the overall RTL specification of the system. Different address Generators
are used for generating the read and write addresses. These addresses are
muxed onto single address interface to the memory element, using the
address mux. Along with the read/write signal (R/$\overline{W}$), the
address is used to write or read the data from the dual port memory unit.
An interface diagram of the entity is shown as in figure
\ref{mem_unit_fig}.

\begin{figure}[h]
\begin{center}
\includegraphics[scale=0.6]{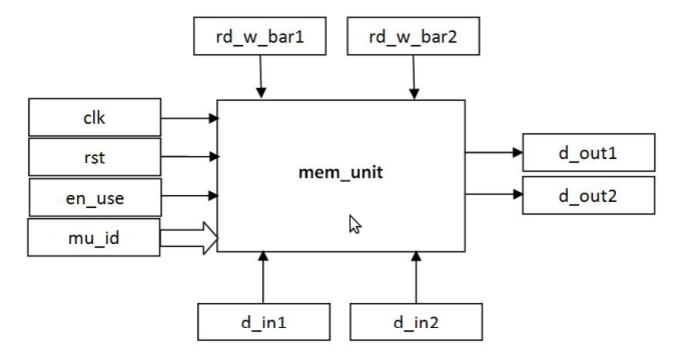}
\end{center}
\caption{Interface Diagram of Memory Unit}
\label{mem_unit_fig}
\end{figure}

The signals used by this entity are the clock, reset, enable, 2
$R/\overline{W}$ signals for two ports of memory element, 2 inputs and 2
outputs, and a memory unit id. The signals that need to be parameterized
are the signal input/output and memory unit id, since their width is
variable. The width of input/output signals depends on the width of
fixed/floating point arithmetic being used for computation. The width of
memory unit id depends on how many memory units are present in the system.
The specific routine in the software that deals with generation of this
entity, takes as its inputs these variable widths. It then generates
formatted outputs based on these parameters, that are the portions affected
by variability of these parameters. For memory units, such portions turn
out to be just part of PORT specification, e.g. \\
           mu\_id     : IN STD\_LOGIC\_VECTOR( value(mu\_width) downto 0 );
           \\

We emit such portions of VHDL model using string formatting
routines. The remaining model, which is unchanged, contains of multiple
pieces of pre-written VHDL files. These files, and the formatted parts of
the model are then appended and sequenced together, before being output as
a single RTL model for the entity.

The complete details of synthesis of all entities within the
system is described in \cite{utk_rep}. The output RTL model of the example
decoder, generated using this tool, was tested for its \textbf{semantic
correctness} using ModelSim 6.6. To also demonstrate the intension that
this model be synthesizable further (i.e., uses only the \textit{synthesis
subset} of VHDL language), we further synthesized it using Xilinx XST tool,
bundled with ISE version 10.1i. The entire tool software is available on
request with authors.

\section{Conclusion}
We have presented a complete design methodology to design folded, pipelined
architectures for applications based on PG bipartite graphs. The underlying
scheme of partitioning is based on simple mathematical concepts, and hence
easy to implement. Usage of this methodology yields static interconnect between
various components, thus saving overheads of switch reconfigurations across
scheduling of various folds. Simple addressing schemes, no switch
reconfiguration etc. lead to ease of implementation, which is another
advantage. The design methodology is based on five levels of model
abstractions, and successive refinement between them. It has a close
correspondence with SpecC based system design methodology, and also
with general SoC design methodologies. It reinforces our belief that
practical, useful design flows can be implemented for this methodology. In
fact, a specific design of an LDPC decoder based on this methodology
was worked out in past \cite{ldpc_foldpat}. Alternate, dual methods of
folding have also been worked out as part of our research theme of folded
architectures \cite{cacs_pap}, \cite{expanders}. Work is ongoing to mould
these partitioning methods into complete alternate design methodologies. Given
the performance advantage of using PG in e.g. design of certain optimal
recent-generation error-correction codes \cite{expanders},
\cite{ldpc_pap}, we believe that such folding methodologies have more
potential scope of application in future.

\bibliography{ref}

\appendix
\section{Projective Spaces as Finite Field Extension}
\label{appA}

This appendix provides an overview of how the projective spaces are generated
from finite fields.  As mentioned before, projective spaces and their
lattices are built using vector subspaces of the \textbf{bijectively}
corresponding vector space, one dimension high, and their subsumption
relations. Vector spaces being extension fields, Galois fields are used to
practically construct projective spaces \cite{expanders}.

Consider a finite field {\large $\mathbb{F}$} = {\large $\mathbb{GF}(s)$}
with {\large $\mathbf{s}$} elements, where {\large
$\mathbf{s}=\mathbf{p^{k}}$}, {\large $\mathbf{p}$} being a prime number
and {\large $\mathbf{k}$} being a positive integer. A projective space of
dimension {\large $\mathbf{d}$} is denoted by {\large
${\mathbb{P}}(d,\mathbb{F})$} and consists of one-dimensional
vector subspaces of
the {\large $(\mathbf{d+1})$}-dimensional vector space over {\large
$\mathbb{F}$} (an extension field over {\large $\mathbb{F}$}), denoted by
{\large $\mathbb{F}^{d+1}$}. Elements of this vector space are denoted by
the sequence {\large $(\mathbf{x_{1},\ldots,x_{d+1}})$}, where each {\large
$\mathbf{x_{i}} \in \mathbb{F}$}. The total number of such elements are
{\large $\mathbf{s^{(d+1)}}$} = {\large $\mathbf{p^{k(d+1)}}$}. An
equivalence relation between these elements is defined as follows. Two
non-zero elements {\large ${\bf{x}}$}, {\large ${\bf{y}}$} are
\textit{equivalent} if there exists an element {\large $\lambda \in$}
{\large $\mathbb{GF}(\mathbf{s})$} such that {\large ${\bf{x}}=\lambda
{\bf{y}}$}. Clearly, each equivalence class consists of {\large
$\mathbf{s}$} elements of the field ({\large $(\mathbf{s-1})$} non-zero
elements and {\large ${\bf{0}}$}), and forms a one-dimensional
vector subspace.
Such 1-dimensional vector subspace corresponds to a \textbf{point} in the projective
space. Points are the zero-dimensional subspaces of the projective space.
Therefore, the total number of points in {\large
${\mathbb{P}}(d,\mathbb{F})$} are

{\large
\begin{equation}
\label{eq1}
P(d) = \frac{s^{d+1}-1}{s-1}
\end{equation}
}

An {\large $\mathbf{m}$}-dimensional projective subspace of {\large
${\mathbb{P}}(d,\mathbb{F})$} consists of all the one-dimensional
vector subspaces contained in an {\large
$(\mathbf{m+1})$}-dimensional subspace of the vector space.
The basis of this vector subspace will have {\large $(\mathbf{m+1})$}
linearly independent elements, say {\large $\mathbf{b_{0},\ldots,b_{m}}$}.
Every element of this vector subspace can be represented as a linear combination
of these basis vectors.
{\large
\begin{equation}
\label{eq2}
{\bf{x}} = \sum_{i=0}^{m} \alpha_{i} b_{i}, \textrm{ where } \alpha_{i} \in
\mathbb{F}(s)
\end{equation}
}

Clearly, the number of elements in the vector subspace are {\large
$\mathbf{s^{(m+1)}}$}.  The number of points contained in the {\large
$\mathbf{m}$}-dimensional projective subspace is given by {\large
$P(\mathbf{m})$} defined in equation (\ref{eq1}). This {\large
$(\mathbf{m+1})$}-dimensional vector subspace and the corresponding
projective subspace are said to have a \textit{co-dimension} of {\large
$\mathbf{r}=\mathbf{(d-m)}$} (the rank of the null space of this vector
subspace). Various properties such as degree etc. of a {\large
$\mathbf{m}$}-dimensional projective subspace remain same, when
this
subspace is bijectively mapped to {\large $(\mathbf{d-m-1})$}-dimensional
projective subspace, and vice-versa. This is known as the \textit{duality
principle} of projective spaces.

An example \textit{Finite Field} and the corresponding Projective Geometry
can be generated as follows. For a particular value of {\large
$\mathbf{s}$} in {\large $\mathbb{GF}$}(s), one needs to first find a
\textit{primitive polynomial} for the field. Such polynomials are
well-tabulated in various literature. For example, for the (smallest)
projective geometry, {\large $\mathbb{GF}$}({\large $2^3$}) is used for
generation. One primitive polynomial for this Finite Field is {\large
$\mathbf{(x^3+x+1)}$}. Powers of the root of this polynomial, {\large
$\mathbf{x}$}, are then successively taken, ({\large $2^3 -1$}) times,
modulo this polynomial, modulo-{\large 2}.  This means, {\large
$\mathbf{x^3}$} is substituted with {\large $(\mathbf{x+1})$}, wherever
required, since over base field {\large $\mathbb{GF}(2)$, -1 = 1}. A
\textit{sequence} of such evaluations lead to generation of the sequence of
{\large $(\mathbf{s-1})$} Finite field elements, \textbf{other than 0}.
Thus, the sequence of {\large $2^3$} elements for {\large
$\mathbb{GF}$}({\large $2^3$}) is \textbf{0(by default)}, {\large $\alpha^0
= 1, \alpha^1 = \alpha, \alpha^2 = \alpha^2, \alpha^3 = \alpha + 1,
\alpha^4 = \alpha^2 + \alpha, \alpha^5 = \alpha^2 + \alpha + 1, \alpha^6 =
\alpha^2 + 1$}.

\begin{figure*}[h]
\centerline{\subfloat[Line-point
        Association]{\includegraphics[scale=0.2]{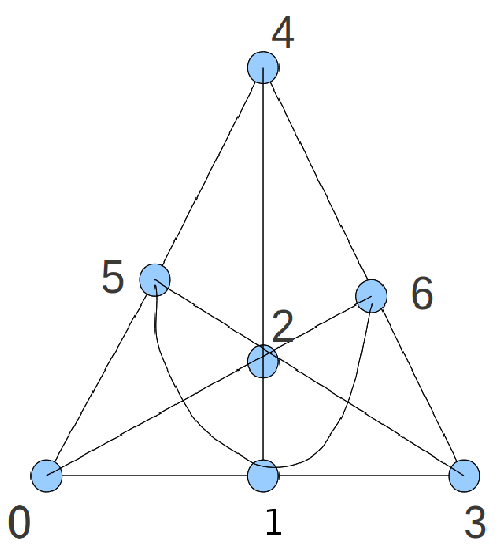}
\label{subfig1}}
\hfil
\subfloat[Bipartite
Representation]{\includegraphics[scale=0.4]{fano_bip}
\label{subfig2}}}
\caption{2-dimensional Projective Geometry}
\label{fano_pic}
\end{figure*}

To generate \textit{Projective Geometry} corresponding to above Galois
Field example({\large $\mathbb{GF}$}({\large $2^3$})), the
2-dimensional projective
plane, we treat each of the above \textit{non-zero} element, the
\textit{lone} non-zero element of various 1-dimensional vector subspaces, as
\uline{points} of the geometry. Further, we pick various subfields(vector
subspaces) of {\large $\mathbb{GF}$}({\large $2^3$}), and label them as
various \uline{lines}. Thus, the seven lines of the projective plane are
\{{\large 1, $\alpha$}, {\large $\alpha^3$} = {\large $1+\alpha$}\},
\{{\large 1, $\alpha^2$}, {\large $\alpha^6$} = {\large $1+\alpha^2$}\},
\{{\large $\alpha$}, {\large $\alpha^2$}, {\large $\alpha^4$} = {\large
$\alpha^2+\alpha$}\}, \{{\large 1,$\alpha^4$} = {\large $\alpha^2+\alpha$},
{\large $\alpha^5$} = {\large $\alpha^2+\alpha+1$}\}, \{{\large $\alpha$},
{\large $\alpha^5$} = {\large $\alpha^2+\alpha+1$}, {\large $\alpha^6$} =
{\large $\alpha^2+1$}\}, \{{\large $\alpha^2$}, {\large $\alpha^3$} =
{\large $\alpha+1$}, {\large $\alpha^5$} = {\large $\alpha^2+\alpha+1$}\}
and \{{\large $\alpha^3$} = {\large $1+\alpha$}, {\large $\alpha^4$} =
{\large $\alpha+\alpha^2$}, {\large $\alpha^6$} = {\large $1+\alpha^2$}\}.
The corresponding geometry can be seen as figures \ref{fano_pic}.

Let us denote the collection of all the {\large $\mathbf{l}$}-dimensional
projective subspaces by {\large $\mathbf{\Omega_{l}}$}. Now, {\large
$\mathbf{\Omega_{0}}$} represents the set of all the points of the
projective space, {\large $\mathbf{\Omega_{1}}$} is the set of all lines,
{\large $\mathbf{\Omega_{2}}$} is the set of all planes and so on. To count
the number of elements in each of these sets, we define the function

{\large
\begin{equation}
\label{eq3}
\phi(n,l,s)=\frac{(s^{n+1}-1)(s^{n}-1)\ldots(s^{n-l+1}-1)}{(s-1)(s^{2}-1)\ldots(s^{l+1}-1)}
\end{equation}
}

Now, the number of {\large $\mathbf{m}$}-dimensional projective subspaces of {\large
${\mathbb{P}}(d,\mathbb{F})$} is {\large $\phi(d,m,s)$}. For example, the
number of points contained in {\large ${\mathbb{P}}(d,F)$} is {\large
$\phi(d,0,s)$}. Also, the number of {\large $\mathbf{l}$}-dimensional projective 
subspaces contained in an {\large $\mathbf{m}$}-dimensional projective subspace (where
{\large $0 \leq l<m \leq d$}) is {\large $\phi(m,l,s)$}, while the number
of {\large $\mathbf{m}$}-dimensional projective subspaces containing a particular
{\large $\mathbf{l}$}-dimensional projective subspace is {\large
$\phi(d-l-1,m-l-1,s)$}.

\end{document}